\documentclass[letterpaper,11pt]{article}

\usepackage{verbatim, url}
\usepackage{latexsym}
\usepackage{amsmath,amssymb}
\usepackage{epsfig}
\usepackage{fullpage}
\usepackage{times}

\usepackage{tikz, fp, subfigure, epsfig}
\usepackage{hyperref, cite}
\usepackage{color}

\newcommand\drop[1]{}

\newtheorem{theorem}{Theorem}
\newtheorem{lemma}[theorem]{Lemma}

\newtheorem{proposition}[theorem]{Proposition}

\newtheorem{definition}[theorem]{Definition}
\newcommand{\qed}{\hbox{\rule{6pt}{6pt}}}
\newenvironment{proof}[1][]{\paragraph{Proof{#1}}}{\hfill\qed\medskip\\}
\newcommand\req[1]{(\ref{#1})}

\newcommand\Prp[1]{\Pr\left[{#1}\right]}
\def\eps{\varepsilon}
\def\argmin{\textnormal{argmin}}

\newcommand\E{{\sf E}}
\newcommand\Var{{\sf Var}}

\newcommand{\hash}{h}
\newcommand{\ul}{\underline}
\newcommand{\ol}{\overline}
\newcommand\kmin{$k\times$min}
\newcommand\botk{bottom-$k$}

\newcommand{\uw}[1]{\underline1w_{{#1}}}

\newcommand{\fw}[1]{fw_{{#1}}}

\newenvironment{description*}%
  {\vspace{-1ex}\begin{description}%
    \setlength{\itemsep}{-0.5ex}%
    \setlength{\parsep}{0pt}}%
  {\end{description}}

\title{Bottom-k and Priority Sampling, Set Similarity and Subset Sums  
with Minimal  Independence\footnote{A short preliminary version of
this paper was presented at STOC'13 \cite{Tho13:bottomk}.}}

\author{Mikkel Thorup\\
AT\&T Labs---Research and University of Copenhagen\\
\texttt{mikkel2thorup@gmail.com}}

\begin{document}
\maketitle
\begin{abstract} We consider {\botk} sampling for a set $X$, picking
a sample $S_k(X)$ consisting of the $k$ elements that are smallest
according to a given hash function $h$. With this sample we can estimate the 
frequency
$f=|Y|/|X|$ of any subset $Y$ as $|S_k(X)\cap Y|/k$. A standard application
is the estimation of the Jaccard similarity $f=|A\cap B|/|A\cup B|$ between 
sets $A$ and $B$. Given the {\botk} samples from $A$ and $B$, we construct
the {\botk} sample of their union as $S_k(A\cup B)=S_k(S_k(A)\cup S_k(B))$, and then the similarity is estimated as $|S_k(A\cup B)\cap S_k(A)\cap S_k(B)|/k$.

We show here that even if the hash function is only 2-independent, the
expected relative error is $O(1/\sqrt{fk})$. For $fk=\Omega(1)$ this
is within a constant factor of the expected relative error with truly random hashing.

For comparison, consider the classic approach of repeated min-wise hashing,  where we use $k$ independent hash functions $h_1,...,h_k$, storing the smallest
element with each hash function. For min-wise hashing, there
can be a constant bias with constant independence, and
this is not reduced with more repetitions $k$. Recently Feigenblat et al.~showed
that {\botk} circumvents the bias if the hash function is 8-independent
and $k$ is sufficiently large. We get down to 2-independence for any
$k$. Our result is based on a simple union bound, transferring
generic concentration bounds for the hashing scheme to the
{\botk} sample, e.g., getting stronger probability error bounds with
higher independence.

For weighted sets, we consider priority sampling which adapts
efficiently to the concrete input weights, e.g., benefiting strongly
from heavy-tailed input. This time, the analysis is much more involved,
but again we show that generic concentration bounds can be applied. 
\end{abstract}

\section{Introduction}

The concept of min-wise hashing (or the ``MinHash algorithm''
according to
\footnote{See \url{http://en.wikipedia.org/wiki/MinHash}} ) is a basic
algorithmic tool suggested by Broder et al.~\cite{Broder97onthe,broder98minwise} for
problems related to set similarity and containment. After the initial
application of this algorithm in the early Altavista search engine to
detecting and clustering similar documents, the scheme has reappeared
in numerous other applications\footnotemark[1] and is now a standard
tool in data mining where it is used for estimating similarity
\cite{broder98minwise,Broder97onthe,Broder97minwise},
rarity \cite{DatarM02estimatingrarity},
document duplicate detection
\cite{Broder00,MJS07,YC06,Henzinger06},
etc  \cite{BHP09,BPR09,CDFGIMUY01,SWA03}. 

In an abstract mathematical view, we have two sets, $A$ and $B$,  and we are interested in understanding their overlap in the
sense of the Jaccard similarity $f = \frac{|A\cap B|}{|A\cup B|}$. In order to
do this by sampling, we need sampling correlated between the two
sets, so we sample by hashing. 
Consider a hash function $h: A\cup B\to
[0,1]$. For simplicity we assume that $h$ is fully random, and has enough
precision that no collisions are expected.
The main mathematical observation is that $\Pr[\argmin h(A) = \argmin h(B)]$ is
precisely $f=|A\cap B|~\big/~|A\cup B|$. Thus, we may sample the element
with the minimal hash from each set, and use
them in $[\argmin\, h(A)=\argmin\, h(B)]$ for an unbiased estimate
of $f$. Here, for a logical statement $S$, $[S]=1$ if $S$ is true; otherwise
$[S]=0$.

For more concentrated estimators, we use repetition with $k$ independent hash functions, $h_1,...,h_k$. For
each set $A$, we store $M^k(A)=(\argmin\, h_1(A),...,\argmin\,
h_k(A))$, which is a sample with replacement from $A$. The Jaccard similarity between sets $A$ and $B$ is now estimated as
$|M^k(A)\cap M^k(B)|/k$ where $|M^k(A)\cap M^k(B)|$ denotes the number of agreeing
coordinates between $M^k(A)$ and $M^k(B)$. 
We shall refer to this approach as repeated min-wise or {\em \kmin}.

For our discussion, we consider the very related application where we wish
to store a sample of a set $X$ that we can use to estimate the
frequency $f=\frac{|Y|}{|X|}$ of any subset $Y\subseteq X$. The idea is
that the subset $Y$ is not known when the sample from $X$ is made. The subset
$Y$ is revealed later in the form of a characteristic function that can tell if (sampled) elements belong to $Y$. Using the
{\kmin} sample $M^k(X)$, we estimate the frequency as $|M^k(X)\cap
Y|/k$ where $|M^k(X)\cap
Y|$ denotes the number of samples from $M^k(X)$ in $Y$. 

Another classic approach for frequency estimation is to use just one
hash function $h$ and use the $k$ elements from $X$ with the smallest
hashes as a sample $S_k(X)$. This is a sample without replacement from
$X$. As in \cite{CK07}, we refer to this as a {\botk} sample. The
method goes back at least to \cite{HHW97}.  The frequency of $Y$ in
$X$ is estimated as $|Y \cap{} S_k(X)|/k$. Even though surprisingly fast
methods have been proposed to compute {\kmin} \cite{BP10}, the bottom-$k$ signature is
much simpler and faster to compute. In a single pass through a set, we only
apply a single hash
function $h$ to each element, and use a max-priority queue to maintain the $k$ smallest
elements with respect to $h$. 

It is standard\footnotemark[1] to use bottom-$k$ samples to
estimate the Jaccard similarity between sets $A$ and $B$, for this
is exactly the frequency of the intersection in the union. First we
construct the bottom-$k$ sample $S_k(A\cup{}B)=S_k(S_k(A)\cup S_k(B))$ of the
union by picking the $k$ elements from $S_k(A)\cup S_k(B)$ with the
smallest hashes. Next we return $|S_k(A)\cap{}S_k(B)\cap{}S_k(A\cup{}B)|/k$.

Stepping back, for subset frequency, we generally assume that we can
identify samples from the subset. In the application to set
similarity, it important that the samples are coordinated via hash
functions, for this is what allows us to identify samples from the
intersection as being sampled in both sets. In our mathematical
analysis we will focus on the simpler case of subset frequency
estimation, but it the application to set similarity that motivates
our special interest in sampling via hash functions.

\paragraph{Limited independence}
The two approaches {\kmin} and {\botk} are similar in spirit, starting
from the same base $1\times$min $=$ bottom-$1$. With truly random hash
functions, they have essentially the same relative standard deviation (standard deviation divided by expectation) bounded by $1/\sqrt{fk}$ where $f$ is the set similarity
or subset frequency. The two approaches are, however, very different
from the perspective of pseudo-random hash functions of limited
independence \cite{wegman81kwise}: a random hash function $h$ is $d$-independent if the hash values of any $d$ given elements are totally random.

With min-wise hashing, we have a problem with bias in the sense of
sets in which some elements have a better than average chance of
getting the smallest hash value. It is known that $1+o(1)$ bias
requires $\omega(1)$-independence \cite{PT10}.  This bias is not
reduced by repetitions as in {\kmin}. However, recently Porat et
al. \cite{FPS12} proved that the bias for bottom-$k$ vanishes for
large enough $k\gg 1$ if we use $8$-independent hashing. Essentially they
get an expected relative error of $O(1/\sqrt{fk})$, and error includes
bias. For $fk=\Omega(1)$, this is only a constant factor worse than
with truly random hashing. Their results are cast in a new framework
of ``$d$-$k$-min-wise hashing'', and the translation to our context is
not immediate.

\paragraph{Results}
In this paper, we prove that bottom-$k$ sampling preserves the
expected relative error of $O(1/\sqrt{fk})$  with
$2$-independent hashing, and this holds for any $k$ including $k=1$.
We note that when $fk=o(1)$, then $1/\sqrt{fk}=\omega(1)$, so our
result does not contradict a possible large bias for $k=1$.

We remark that we also get an $O(1/\sqrt{(1-f)k})$
bound on the expected relative error. This is important if we estimate 
the dissimilarity $1-f$ of sets with large similarity $f=1-o(1)$.

For the more general case of weighted sets, we consider priority
sampling \cite{DLT07:priority} which adapts near-optimally to the
concrete input weights \cite{Sze06}, e.g., benefiting strongly from
heavy-tailed input. We show that 2-independent hashing suffices for
good concentration.

Our positive finding with 2-independence contrasts recent negative results on the
insufficiency of low independence, e.g., that linear probing needs the
5-independence \cite{PT10} that was proved sufficient by Pagh et
al. \cite{pagh07linprobe}.

\paragraph{Implementation}
For 2-independent hashing we can use the fast multiplication-shift
scheme from \cite{Die96}, e.g., if the elements are 32-bit keys, we
pick two random 64-bit numbers $a$ and $b$. The hash of key $x$ is
computed with the C-code $(a*x+b)>> 32$, where $*$ is 64-bit
multiplication which as usual discards overflow, and $>>$ is a right shift.
This is 10-20 times faster than the fastest known $8$-independent
hashing based on a degree 7 polynomial tuned for a Mersenne prime
field \cite{TZ12}\footnote{See Table 2 in \cite{TZ12} for comparisons
  with different key lengths and computers between
  multiplication-shift (TwoIndep), and tuned polynomial hashing
  (CWtrick). The table considers polynomials of degree 3 and 4, but
  the cost is linear in the degree, so the cost for degree 7 is easily
  extrapolated.}.

\paragraph{Practical relevance}
We note that Mitzenmacher and Vadhan \cite{mitzenmacher08hash} have
proved that 2-independence generally works if the input has enough
entropy.  However, the real world has lots of low entropy data. In
\cite{TZ12} it was noted how consecutive numbers with zero entropy made
linear probing with 2-independent hashing extremely unreliable. This
was a problem in connection with denial-of-service attacks 
using consecutive IP-addresses. For our set similarity, we would have
similar issues in scenarios where small numbers are more common, hence
where set intersections are likely to be fairly dense intervals of small
numbers whereas the difference is more likely to consists of large random
outliers. Figure \ref{fig:exper} presents an experiment showing what
happens if we try to estimating such dissimilarity with 2-independent hashing.
\begin{figure*}
\begin{center}
\begin{tabular}{c}
\kmin\\[-.5ex]
\hspace{-.2in}\includegraphics[width=3.1in]{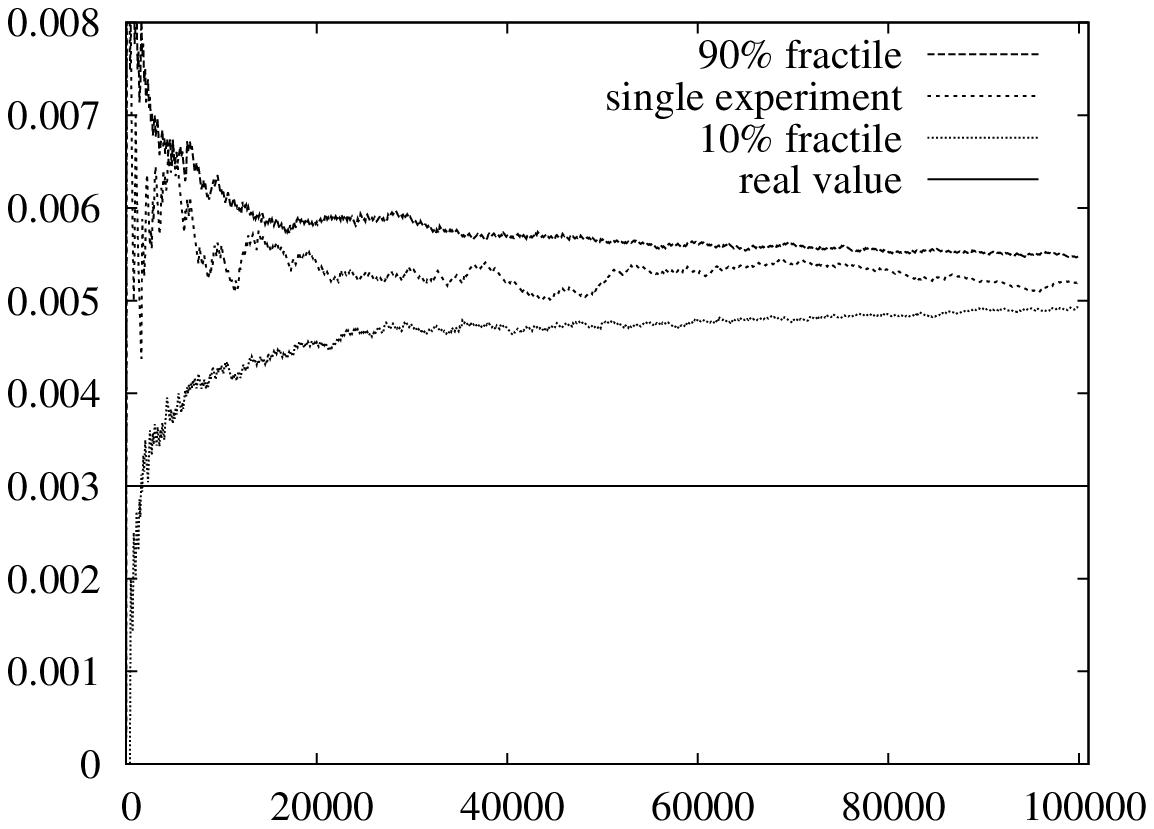}\\[-4ex]
\end{tabular}~\begin{tabular}{c}
\botk\\[-.5ex]
\includegraphics[width=3.1in]{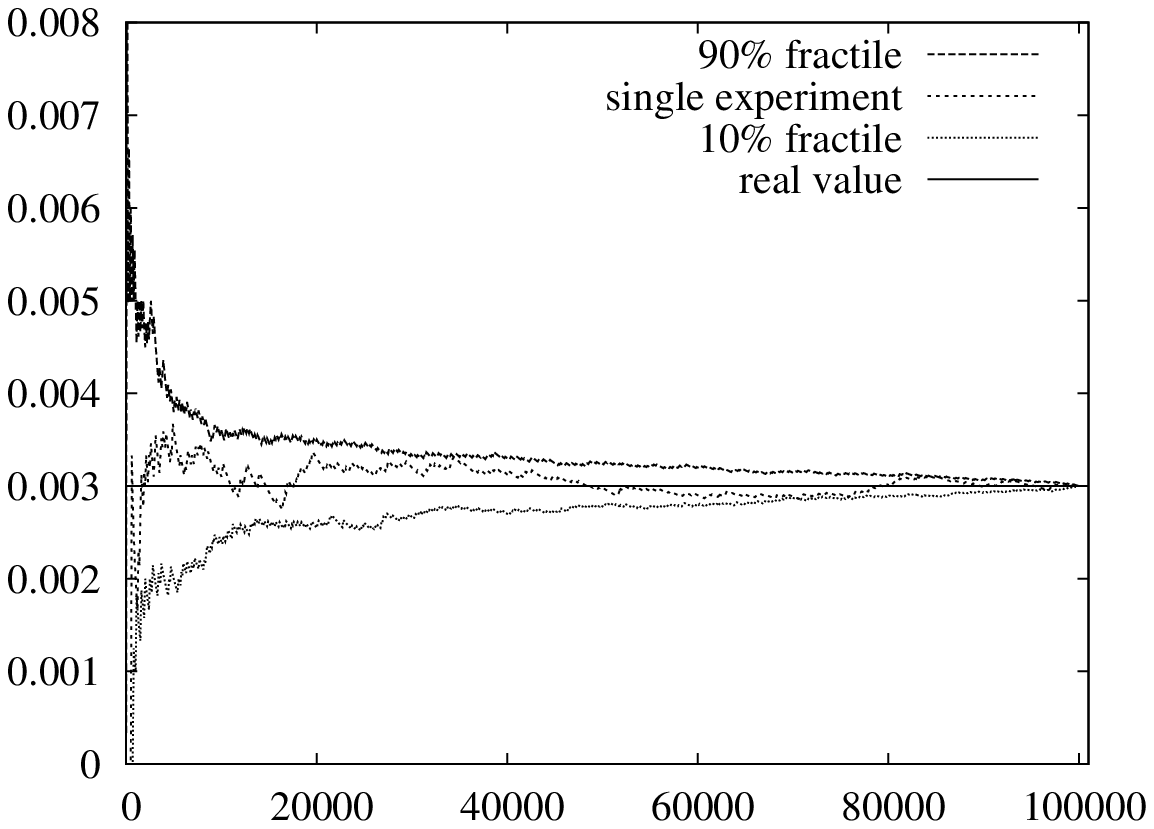}\\[-4ex]
\end{tabular}
\end{center}
\caption{Experiment with set consisting of 100300 32-bit keys.  It has
  a ``core'' consisting of the consecutive numbers $1,...100000$. In addition it
  has  $300$ random ``outliers''.  Using $k$ samples from the whole set,
  we want to estimate the frequency of the outliers. The true frequency is
  $\frac{300}{100300}\approx 0.003$. We used $k=1,...,100000$ in
  {\kmin} and {\botk} and made one hundred experiments.  For each $k$,
  we sorted the estimates, plotting the 10th and 90th value, labeled
  as 10\% and 90\% fractile in the figures. We also plotted the
  results from a single experiment. For readability, only one in every
  100 values of $k$ is plotted. Both schemes converge, but due to
  bias, {\kmin} converges to a value that is $70\%$ too large. Since
  {\botk} does sampling without replacement, it becomes exact when the
  number of samples is the size of the whole set.  The bias is a
  function of the structure of the subset within the whole set, e.g.,
  the core set must have a negative bias complimenting the positive
  bias of the outliers.  It is therefore not possible to correct for
  the bias if one only has the sample available.  }\label{fig:exper}
\end{figure*}

Stepping back, the result Mitzenmacher and Vadhan is that
2-independence works for sufficiently random input. In particular, we
do not expect problems to show up in random tests. However, this does
not imply that 2-independent hashing can be trusted on real data
unless we have specific reasons to believe that the input has high
entropy. In Figure~\ref{fig:exper}, {\botk} performs beautifully with
2-independent hashing, but no amount of experiments can demonstrate
general reliability. However, the mathematical result of this paper is
that {\botk} can indeed be trusted with 2-independent hashing: the
expected relative error is $O(1/\sqrt{fk})$ no matter the structure of
the input.

\paragraph{Techniques}
To appreciate our analysis, let us first consider the trivial case
where we are given a non-random threshold probability $p$ and sample
all elements that hash below $p$. As in \cite{DLT05} we refer to this as
{\em threshold sampling}. Since the hash of a element $x$ is uniform in
$[0,1]$, this samples $x$ with probability $p$. The sampling of $x$
depends only on the hash value of $x$, so if, say, the hash function
is $d$-independent, then the number of samples is the sum of
$d$-independent 0-1 variables.  This scenario is very well understood
(see, e.g., \cite{DGMP92,schmidt95chernoff}).

We could set $p=k/n$, and
get an expected number of $k$ samples. Morally, this should be
similar to a bottom-$k$ sample, which is what 
we get if we end up with exactly $k$ samples, that is,
if we end up with $h_{(k)}<p\leq h_{(k+1)}$ where
$h_{(i)}$ denotes the $i$th smallest hash value. What complicates the
situation is that $h_{(k)}$ and $h_{(k+1)}$ are random variables depending 
on all the random hash values.

An issue with threshold sampling is that the number of samples is variable. This
is an issue if we have bounded capacity to store the samples. With
$k$ expected samples, we could put some limit $K\gg k$ on the number of
samples, but any such limit introduces dependencies that have to be understood.
Also, if we have room for $K$ samples, then it would seem wasteful not 
fill it with a full bottom-$K$ sample.

Our analysis of bottom-$k$ samples is much simpler than the one in
\cite{FPS12} for $8$-independent hashing with $k\gg 1$. With a union bound we reduce
the analysis of bottom-$k$ samples to the trivial case of threshold
sampling.  Essentially we only get a constant loss in the error
probabilities. With 2-independent hashing, we then apply Chebyshev's
inequality to show that the expected relative error is
$O(1/\sqrt{fk})$. The error probability bounds are immediately
improved if we use hash functions with higher independence.

It is already known from \cite{BJKST02} that we can use a
2-independent bottom-$k$ sample of a set to estimate its size $n$ with
an expected error of $O(\sqrt{n})$. The estimate is simply the inverse of
the $k$th smallest sample. Applying this to two $\Theta(n)$-sized sets
and their union, we can estimate $|A|$, $|B|$, $|A\cup B|$ and $|A\cap
B|=|A|+|B|-|A\cup B|$ each with an expected error of $O(\sqrt{n})$.
However, $|A\cap B|$ may be much smaller than $O(\sqrt{n})$.
If we instead multiply our estimate of the similarity $f=|A\cap
B|/|A\cup B|$ with the estimate of $|A\cup B|$,
the resulting estimate of $|A\cap B|$ is
\[(1\pm O(1/\sqrt{fk})f(|A\cup B|\pm O(\sqrt{n}))=|A\cap B|\pm O(\sqrt{|A\cap B}).\]

The analysis of
priority sampling for weighted sets is much more delicate,
but again, using union bounds, we show that generic concentration bounds apply.

\section{Bottom-k samples}
We are given a set of $X$ of $n$ elements. A hash function maps
the elements uniformly and collision free into $[0,1]$. Our
bottom-$k$ sample $S$ consists of the $k$ elements with the lowest
hash values. The sample is used to estimate the frequency $f=|Y|/|X|$ of any subset $Y$ of $X$ as $|Y\cap S|/k$.  With 2-independent hashing, we will prove the
following error probability bound for any $r\leq \bar r=\sqrt k/3$:
\begin{equation}\label{eq:2-ind}
\Prp{||Y \cap S| - fk| > r\sqrt{fk}} \leq 4/r^2.
\end{equation}
The result is obtained via a simple union bound where stronger hash functions
yield better error probabilities. 
With $d$-independence with $d$ an even constant, the probability bound
is $O(1/r^d)$. 

It is instructive to compare $d$-independence with the idea of storing $d$ independent {\botk} samples, each based on 2-independence, and use the median estimate. Generally, if the probability of a certain deviation is $p$, the deviation probability for the median is bounded by $(2ep)^{d/2}$, so the $4/r^2$ from \req{eq:2-ind} becomes
$(2e4/r^2)^{d/2}<(5/r)^d$, which is the same type of probability that we get with
a single $d$-independent hash function. The big advantage of a single $d$-independent hash function is that we only have to store a single {\botk} sample.

If we are willing to use much more space for the hash function, then
we can use twisted tabulation hashing \cite{PT13:twisted} which is
very fast, and then we get exponential decay in $r$ though only down
to an arbitrary polynomial of the space used.

In order to show that the expected relative error is $O(1/\sqrt{fk})$, 
we also prove the following bound for $fk\leq 1/4$:
\begin{equation}\label{eq:small-fk}
\Pr[|Y\cap S|\geq\ell]=O(fk/\ell^2+\sqrt f /\ell).
\end{equation}
From \req{eq:2-ind} and \req{eq:small-fk}  we get
\begin{proposition} For bottom-$k$ samples based on 2-independent hashing,
a fraction $f$ subset is estimated with an  expected relative error of 
$O(1/\sqrt{fk})$.
\end{proposition} 
\begin{proof}
The proof assumes \req{eq:2-ind} and \req{eq:small-fk}.
For the case $fk>1/4$, we will apply \req{eq:2-ind}.
The statement is equivalent to saying that the sample error 
$||Y \cap S| - fk|$ in expectation is bounded by $O(\sqrt{fk})$.
This follows immediately from \req{eq:2-ind} for
errors below $\bar r\sqrt{fk}=k\sqrt{|Y|/n}$. However,
by \req{eq:2-ind}, the probability of a larger error
is bounded by $4/{\bar r}^2=O(1/k)$. The maximal error is
$k$, so the contribution of larger errors to the expected error is $O(1)$. 
This is $O(\sqrt{fk})$ since $fk> 1/4$.

We will now handle the case $fk\leq 1/4$ using \req{eq:small-fk}.
We want to show that the expected absolute error is $O(\sqrt{fk})$.
We note that only positive errors can be bigger than $fk$, so
if the expected error is above $2fk$, the expected number of
samples from $Y$ is proportional to the expected error. 
We have $\sqrt{fk}\geq 2fk$, so for the expected error bound, 
it suffices to prove that the
expected number of samples is $|Y\cap S|=O(\sqrt{fk})$. Using 
\req{eq:small-fk} for the probabilities, we now
sum the contributions over exponentially increasing sample sizes.
\begin{align*}
\E[|Y\cap S|]&\leq\sum_{i=0}^{\lfloor \lg k\rfloor}
\left(2^{i+1} \Pr[|Y\cap S|\geq 2^i]\right)\\
&=\sum_{i=0}^{\lfloor \lg k\rfloor}O\left(2^i (fk/2^{2i}+\sqrt f /2^i)\right)\\
&=O\left(fk+\sqrt f(1+\lg k)\right)=O\left(\sqrt{fk}\right).
\end{align*}
\end{proof}

\subsection{A union upper bound}\label{sec:upper}
First we consider overestimates. 
For positive parameters $a$ and $b$ to be chosen, we will bound the probability 
of the overestimate
\begin{equation}\label{*}
   |Y \cap{} S|> \frac{1+b}{1-a}\,fk.
\end{equation}
Define the threshold probability
\[p=\frac k{n(1-a)}.\]
Note that $p$ is defined deterministically, independent of any
samples. It is easy to see that the overestimate \req{*} implies one
of the following two threshold sampling events:
\begin{description}
\item[(A)] The number of elements from $X$ that hash below $p$ is less
than $k$.  We expected $pn=k/(1-a)$ elements, so $k$ is a factor $(1-a)$ below the
expectation.
\item[(B)] $Y$ gets more than $(1+b)p|Y|$ hashes below $p$, that is,
a factor $(1+b)$ above the expectation.
\end{description}
To see this, assume that both (A) and (B) are false. When (A) is
false, we have $k$ hashes from $X$ below $p$, so the largest hash
in $S$ is below $p$. Now if (B) is also false, we have at most
$(1+b)p|Y|=(1+b)/(1-a)\cdot fk$ elements from $Y$ hashing below $p$, 
and only these elements from $Y$ could be in 
$S$. This contradicts \req{*}. By the union bound, we have proved
\begin{proposition}\label{union} The probability of the overestimate \req{*}
is bounded by $P_A+P_B$ where $P_A$ and $P_B$ are the probabilities
of the events (A) and (B), respectively.
\end{proposition}

\paragraph{Upper bound with 2-independence}\label{2-ind}
Addressing events like (A) and (B), let $m$ be the number of elements in
the set $Z$ considered, e.g., $Z=X$ or $Z=Y$.  We study the number of
elements hashing below a given threshold $p\in[0,1]$.  Assuming that the
hash values are uniform in $[0,1]$, the mean is $\mu=mp$. Assuming
2-independence of the hash values, the variance is $mp(1-p)=(1-p)\mu$ and
the standard deviation is $\sigma=\sqrt{(1-p)\mu}$. By Chebyshev's
inequality, we know that the probability of a deviation by $r\sigma$
is bounded by $1/r^2$.  Below we will only use that the relative standard
deviation $\sigma$ bounded by $1/\sqrt\mu$.

For any given $r\leq \sqrt k/3$, we will fix $a$ and $b$ to give a combined
error probability of $2/r^2$. More precisely, we will fix
$a=r/\sqrt k$ and $b=r/\sqrt{fk}$. This also fixes 
$p=k/(n(1-a))$. We note for later that $a\leq 1/3$ and $a\leq b$. This
implies 
\begin{equation}\label{eq:ab}
(1+b)/(1-a)\leq (1+3b)=1+3r/\sqrt{fk}.
\end{equation}
In connection with (A) we study the number of elements from $X$ hashing 
below $p$. The mean is $pn\geq k$ so the relative standard deviation is
less than $1/\sqrt k$. It follows that a relative error of $a=r/\sqrt k$
corresponds to at least $r$ standard deviations, so 
\[P_A=\Prp{\#\{x\in X|h(x)<p\}<(1-a)np}<1/r^2.\]
In connection with (B) we study the number of elements from $Y$ hashing 
below $p$. Let $m=|Y|$. The mean is $pm=km/(n(1-a))$ and the relative standard deviation less than 
$1/\sqrt{pm}<1/\sqrt{km/n}$. 
It follows than a relative error of $b=r/\sqrt{km/n}$
is more than $r$ standard deviations, so 
\[P_B=\Prp{\#\{y\in Y|h(y)<p\}>(1+b)mp}<1/r^2.\]
By Proposition \ref{union} we conclude that the probability of 
\req{*} is bounded by $2/r^2$. Rewriting \req{*} with \req{eq:ab},
we conclude that 
\begin{equation}\label{eq:const-bounds}
\Prp{|Y \cap S| > fk+3r\sqrt{fk}} \leq  2/r^2.
\end{equation}
This bounds the probability of the positive error in \req{eq:2-ind}.
The above constants $3$ and $2$ are moderate, and they can easily be
improved if we look at asymptotics. Suppose we want good
estimates for subsets $Y$ of frequency at least $f_{\min}$, that is,
$|Y|\geq f_{\min}|X|$. This time, we set $a=r/\sqrt{f k}$, and then we get
$P_A\leq f/r^2$. We also set $b=r/\sqrt{fk}$ preserving $P_B\leq
1/r^2$. Now for any $Y\subseteq X$ with $|Y|>fn$, we have
\begin{align}\label{eq:f}
&\Prp{|Y \cap S| > (1+\eps)fk} = (1+f)/r^2\\
&\textnormal{ where }\eps=
\frac{1+r/\sqrt{fk}}{1-r/\sqrt k}-1=\frac{r/\sqrt k+r/\sqrt{fk}}{1-r/\sqrt k}.
\nonumber
\end{align}
With $f=o(1)$ and $k=\omega(1)$, the error is 
$\eps=(1+o(1))r/\sqrt{fk}$, and the error probability is 
$P_\eps=(1+f)/r^2=(1+o(1))/r^2$. Conversely, this means that
if we for subsets of frequency $f$ and a relative positive 
error $\eps$ want an error probability 
around $P_\eps$, then we set $r=\sqrt{1/P_\eps}$ and $k=r^2/(f\eps^2)=
1/(f\,P_\eps\,\eps^2)$.

\subsection{A union lower bound}\label{sec:lower}
We have symmetric bounds for underestimates:
\begin{equation}\label{*-}
   |Y \cap{} S|< \frac{1-b'}{1+a'}\,fk.
\end{equation}
This time we define the threshold probability $p'=\frac k{n(1+a')}$. It is easy to see that
the overestimate \req{*} implies one of the following two events:
\begin{description}
\item[(A$'$)] The number of elements from $X$ below $p'$ is at least
$k$.  We expected $p'n=k/(1+a')$ elements, so $k$ is a factor $(1+a')$ above the
expectation.
\item[(B$'$)] $Y$ gets less than $(1-b')p|Y|$ hashes below $p'$, that is,
a factor $(1-b')$ below the expectation.
\end{description}
To see this, assume that both (A$'$) and (B$'$) are false. When (A$'$) is
false, we have less than $k$ hashes from $X$ below $p'$, so $S$ must contain
all hashes below $p'$. Now if (B) is also false, we have at least
$(1-b)p'|Y|=(1-b)/(1+a)\cdot fk$ elements from $Y\subseteq X$ 
hashing below $p'$, hence which must be in $S$. This
contradicts \req{*-}. By the union bound, we have proved
\begin{proposition}\label{union-} The probability of the underestimate \req{*-}
is bounded by $P_{A'}+P_{B'}$ where $P_{A'}$ and $P_{B'}$ are the probabilities of the events (A$'$) and (B$'$), respectively.
\end{proposition}
\paragraph{Lower bound with 2-independence}\label{2-ind-}
Using Proposition \ref{union-} we will bound the probability of
underestimates, complementing our previous probability bounds for
overestimates from Section \ref{2-ind}.
We will provide bounds for the same overall relative error as we did for the
overestimates; namely
\[\eps=\frac{1+b}{1-a}-1=(a+b)/(1-a)\]
However, for the events (A$'$) and (B$'$) we are going to scale up the
relative errors by a factor $(1+a)$, that is, we will use $a'=a(1+a)$
and $b'=b(1+a)$. The overall relative negative error from \req{*-} is then
\begin{align*}
\eps'&=1-\frac{1-b'}{1+a'}=(a'+b')/(1+a')\\
&<(1+a)(a+b)/(1+a')<(a+b)<\eps.
\end{align*}
Even with this smaller error, we will get better probability bounds than those
we obtained for the overestimates. For (A) we used $1/\sqrt k$ as an
upper bound on the relative standard deviation, so a relative error of
$a$ was counted as $s_A=a\sqrt k$ standard deviations.  In
(A$'$) we have mean $\mu'=np'=k/(1+a')$, so the relative standard
deviation is bounded by $1/\sqrt{k/(1+a')}=\sqrt{1+a+a^2}/\sqrt k$.
This means that for (A$'$), we can count a relative error of $a'=a(1+a)$ as
\begin{align*}
s'_A&=a(1+a)\sqrt k/\sqrt{1+a+a^2}\\&=s_A(1+a)/\sqrt{1+a+a^2}>s_A
\end{align*}
standard deviations. In Section \ref{2-ind} we bounded $P_A$ by $1/s_A^2$,
and now we can bound $P_{A'}$ by $1/{s'_A}^2\leq 1/s_A^2$.  
The scaling has the same positive effect on our probability
bounds for (B$'$). That is, in Section \ref{2-ind}, a relative error
of $b$ was counted as $s_B=b\sqrt{fk}$ standard deviations. With (B$'$) our
relative error of $b'=b(1+a)$ is counted as
\begin{align*}
s'_B&=b(1+a)\sqrt {fk}/\sqrt{1+a+a^2}\\
&=s_B(1+a)/\sqrt{1+a+a^2}>s_B
\end{align*}
standard deviations, and then we can bound $P_{B'}$ by $1/{s'_B}^2\leq 1/s_B^2$.
Summing up, our negative relative error $\eps'$ is smaller than our
previous positive error $\eps$, and our overall negative error probability 
bound $1/{s'_A}^2+1/{s'_B}^2$ is smaller than our previous positive error
probability bound $1/{s_A}^2+1/{s_B}^2$. We therefore translate
\req{eq:const-bounds} to
\begin{equation}\label{eq:const-bounds-}
\Prp{|Y \cap S| < fk-3r\sqrt{fk}} \leq 2/r^2.
\end{equation}
which together with \req{eq:const-bounds} establishes \req{eq:2-ind}.
Likewise \req{eq:f} translates to 
\begin{align}\label{eq:f*}
&\Prp{||Y \cap S| -fk |> \eps fk)} 
\leq 2(1+f)/r^2\\
&\textnormal{ where }\eps=
\frac{1+r/\sqrt{fk}}{1-r/\sqrt k}-1.\nonumber
\end{align}
As for the positive error bounds we note that with $f=o(1)$ and $k=\omega(1)$, the error is $\eps=(1+o(1))r/\sqrt{fk}$ and the error probability is 
$P_\eps=(2+o(1))/r^2$. Conversely, this means that
if we for a target relative error $\eps$ want an error probability 
around $P_\eps$, then we set $r=\sqrt{2/P_\eps}$ and $k=r^2/(f\eps^2)=
2/(f\,P_\eps\,\eps^2)$.

\subsection{Rare subsets}
We now consider the case where the expected number $fk$ of samples from $Y$ is
less than $1/4$. We wish to prove \req{eq:small-fk}
\begin{equation*}
\Pr[|Y\cap S|\geq\ell]=O(fk/\ell^2+\sqrt f /\ell).
\end{equation*}
For some balancing parameter $c\geq 2$, we use the threshold probability 
$p=ck/n$. The error event (A) is that less than $k$ elements from $X$
sample below $p$. The error event (B) is that at least $\ell$ elements
hash below $p$. As in Proposition \ref{union}, we observe that
$\ell$ bottom-$k$ samples from $Y$ implies (A) or (B), hence that
$\Pr[|Y\cap S|\geq\ell]\leq P_A+P_B$.

The expected number of elements
from $X$ that hash below $p$ is $ck$. The error event 
(A) is that we get less than $k$, which is less
than half the expectation. This amounts to at least $\sqrt{ck}/2$
standard deviations, so by Chebyshev's inequality, the probability
of (A) is $P_A\leq 1/(\sqrt{ck}/2)^2=4/(ck)$. 

The event (B) is that at least $\ell$ elements from $Y$
hash below $p$, while the expectation is only $fck$. Assuming that
$\ell\geq 2fck$, the error is by at
least $(\ell/2)/\sqrt{fck}$ standard deviations. By Chebyshev's inequality, the probability
of (B) is $P_B\leq 1/((\ell/2)/\sqrt{fck})^2=4fck/\ell^2$.
Thus
\[P_A+P_B\leq 4/(ck)+4fck/\ell^2.\]
We wish to pick $c$ for balance, that is, 
\[4/ck=4fck/\ell^2\iff c=\ell/(\sqrt f k)\]
However, we have assumed that $c\geq 2$ and that $\ell\geq 2fck$.
The latter is satisfied because $2fck=2fk \ell/(\sqrt f k)=2\sqrt f\ell$
and $f\leq 1/4$. Assuming that $c=\ell/(\sqrt f k)\geq 2$, we get 
\[P_A+P_B\leq 8/(k(\ell/(\sqrt f k)))=8\sqrt f/\ell.\]
When $\ell/(\sqrt f k)<2$, we set $c=2$. Then
\[P_A+P_B\leq 2/k+8fk/\ell^2\leq 16fk/\ell^2.\]
Again we need to verify that
$\ell\geq fck=2fk$, but that follows because $\ell\geq 1$ and $fk\leq 1/4$.
We know that at one of the above two cases applies, so we 
conclude that
\[P[|Y\cap S|\geq\ell]\leq P_A+P_B=O(fk/\ell^2+\sqrt f /\ell)\textnormal,\]
completing the proof of \req{eq:small-fk}.

\section{Priority sampling}
We now consider the more general situation where we are dealing with
a set $I$  of weighted items with $w_i$ denoting the weight of item $i\in I$.
Let $\sum I=\sum_{i\in I}w_i$ denote the total weight of set 
$I$. 

Now that we are dealing with weighted items, we will use {\em priority
  sampling\/} \cite{DLT07:priority} which generalizes the bottom-$k$
samples we used for unweighted elements. Each item or element $i$ is identified by a unique key
which is hashed uniformly to a random number $\hash_i\in (0,1)$.  The
item is assigned a {\em priority\/} $q_i=w_i/\hash_i>w_i$.  In
practice, hash values may have some limited precision $b$, but we assume that $b$ is large enough that the resulting
rounding can be ignored. We assume that all priorities end up distinct
and different from the weights. If not, we could break ties based on
an ordering of the items.  The priority sample $S$ of size $k$
contains the $k$ samples of highest priority, but it also stores a
threshold $\tau$ which is the $(k+1)$th highest priority. Based on
this we assign a weight estimate $\widehat w_i$ to each item $i$. If
$i$ is not sampled, $\widehat w_i=0$; otherwise $\widehat
w_i=\max\{w_i,\tau\}$. A basic result from \cite{DLT07:priority} is
that $\E[\widehat w_i]=w_i$ if the hash function is truly random (in
\cite{DLT07:priority}, the $h_i$ were described as random numbers, but
here they are hashes of the keys).

We note that priority sampling generalize the bottom-$k$ sample we
used for unweighted items, for if all weights are unit, then the $k$
highest priorities correspond to the $k$ smallest hash values. In
fact, priority sampling predates \cite{CK07}, and \cite{CK07}
describes bottom-$k$ samples for weighted items as a generalization
of priority sampling, picking the first $k$ items according to an
arbitrary randomized function of the weights.

The original objective of priority sampling \cite{DLT07:priority} was
subset sum estimation.  A subset $J\subseteq I$ of the items is
selected, and we estimate the total weight in the subset as $\widehat
w_J=\sum\{\widehat w_i|i\in J\cap S\}$. By linearity of expectation,
this is an unbiased estimator.  A cool application from
\cite{DLT07:priority} was that as soon as the signature of the Slammer
worm \cite{Slammer} was identified, we could inspect the priority samples from
the past to track its history and identify infected hosts. An
important point is that the Slammer worm was not known when the
samples were made. Samples are made with no knowledge about which subsets will later turn out to be of interest.

Trivially, if we want to estimate the relative subset weight $\sum
J/\sum I$ and we do not know the exact total, we can divide $\widehat w_J$
with the estimate $\widehat w_I$ of the total. As with the bottom-$k$
sampling for unweighted items, we can easily use priority sampling to estimate
the similarity of sets of weighted items: given the priority sample from two sets, we construct the priority sample of their union, and estimate the
intersection as a subset. 
This is where it is important that we use a hash
function so that the sampling from different sets is coordinated,
e.g., we could not use iterative sampling procedures like the one in
\cite{CDKLT11:varopt}. In the case of histogram similarity, it is natural to
allow the same item to have different weights in different sets.  More
specifically, allowing zero weights, every possible item has a weight
in each set.  For the similarity we take the sum of the minimum weight
for each item, and divide it by the sum of the maximum weight for each
item.  This requires a special sampling that we shall return to in Section \ref{sec:histogram}.

Priority sampling is not only extremely easy to implement on-line with
a standard min-priority queue; it also has some powerful universal
properties in its adaption to the concrete input weights. As conjectured in 
\cite{DLT07:priority} and proved in
\cite{Sze06}, given one extra sample, priority sampling has smaller
variance sum $\sum_i \Var[\widehat w_i]$ than {\em any\/} off-line
unbiased sampling scheme tailored for the concrete input weights. In
particular, priority sampling benefits strongly if there are dominant
weights $w_i$ in the input, estimated precisely as $\widehat
w_i=\max\{w_i,\tau\}=w_i$. 
In the important case of heavy tailed input
distributions \cite{AFT98}, we thus expect most of the total weight to
be estimated without any error.  The quality of a priority sample is
therefore often much better than what can be described in terms of
simple parameters such as total weight, number of items, etc. The
experiments in \cite{DLT07:priority} on real and synthetic data show
how priority sampling often gains orders of magnitude in estimate
quality over competing methods.

The quality of a priority estimate depends completely on the distribution of weights in the input, and often we would like to know how much we can trust a given estimate.
What we really want from a sample is not just an estimate of a subset sum, 
but a confidence interval \cite{Tho06:conf}: from the information in the sample, we want to derive lower and upper bounds that capture the true
value with some desired probability. Some applications of such concervative
bounds are given in \cite{DLT01}.

What makes priority sampling tricky to analyze is that the priority
threshold $\tau$ is a random variable depending on all the random
priorities.  It may be very likely that the threshold $\tau$ ends up
smaller than some dominant weight $w_i$, but it could also be bigger, so
we do have variance on all weight estimates $\widehat w_i$.

All current analysis of priority sampling
\cite{DLT07:priority,Sze06,Tho06:conf} is heavily based on true
randomness, assuming that the priorities are independent random
variables, e.g., the unbiasedness proof from \cite{DLT07:priority}
that $\E[\widehat w_i]=w_i$ starts by fixing the priorities $q_j$ of
all the other items $j\neq i$. However, in this paper, we want to use
hash functions with independence as low as 2, and then any such
analysis breaks down. In fact, bias may now be introduced. To see this, consider the following extreme case of
2-independent hashing of $n$ keys: divide $(0,1]$ into $n$
  subintervals $I_i=(i/n,(i+1)/n]$. With probability $1-1/n$, the keys
are all mapped to different random subintervals, and with probability $1/n$, all
  keys are mapped to the same random subinterval. Within the
  subintervals, the hashing is totally random. This scheme is clearly
  2-independent, but highly restricted for $n>2$. As a simple example
  of bias, consider a priority sample of $k=2$ out of $n=3$ unit
  weight keys. A messy computer calculation shows that the expected
  weight estimates are 1.084.

\paragraph{Relation to threshold sampling}
Generalizing the pattern for unweighted sets, our basic goal is to
relate the error probabilities with priority sampling to the much
simpler case of threshold sampling for weighted items.
In threshold sampling, we are not given a predefined
sample size. Instead we are given a fixed threshold $t$.  We use
exactly the same random priorities as in priority sampling, but now an item is
sampled if and only if $q_i>t$. The weight estimate is
\begin{equation}\label{eq:threshold-estimate}
\widehat w_i^{\,t}=\left\{\begin{array}{ll}
0 & \textnormal{if }q_i\leq t\\
\max\{w_i,t\} & \textnormal{if }q_i>t
\end{array}\right.\end{equation}
In statistics, threshold sampling is known as  Poisson sampling with probability proportional to size 
\cite{SSW92}. The name threshold sampling is taken from \cite{DLT05}. 

The $\widehat w_i^{\,t}$ notation from \req{eq:threshold-estimate} is well-defined
also when $t$ is a variable, and if priority sampling leads to threshold $\tau$, then
the priority estimate for item $i$ is $\widehat w_i=\widehat w_i^\tau$.

With a fixed threshold $t$, it is trivial to see that the estimates
are unbiased, that is, $\E[\widehat w_i^{\,t}]=w_i$; for if $w_i\geq
t$, we always have $\widehat w_i^{\,t}=w_i$, and if $w_i<t$, then
\[\E[\widehat w_i^{\,t}]=t\Pr[q_i>t]=t \Pr[h_i<w_i/t]=w_i.\]
The unbiasedness with fixed threshold $t$ only requires that each
$h_i$ is uniform in $(0,1)$. No independence is required. This contrasts the bias we may get with the variable priority threshold $\tau$ with limited dependence.

With threshold sampling, concentration bounds for subset sum estimates are easily derived. For a subset $J\subseteq I$, the threshold estimate
$\sum_{i\in J} \widehat w_i$
is naturally divided in an exact part for large weights and a variable
part for small weights:
\begin{eqnarray}
\sum_{i\in J, w_i\geq t} \widehat w_i^{\,t}&=&\sum_{i\in J, w_i\geq t} w_i\nonumber\\
\sum_{i\in J, w_i<t} \widehat w_i^{\,t}&=&t\sum_{i\in J, w_i< t} X_i
\textnormal{, where }X_i=[h_i<w_i/t]\in \{0,1\}.\label{eq:subsetvar}
\end{eqnarray}
Each $X_i$ depends on $h_i$ only, so if the 
$h_i$ are $d$-independent, then so are the $h_i$. Let $X=\sum_{i\in J, w_i< t} X_i$ and
$\mu=\E[X]=\sum_{i\in J, w_i<t} w_i/t$. 
As in the unweighted case, if the hash
function is 2-independent, then $\Var[X]<\mu$, and by Chebyshev's inequality
$\Pr[|X-\mu|\geq r\sqrt\mu]\leq 1/r^2$.

Informally speaking, for bounded errors and modulo constant factors,
our main result is that concentration bounds for threshold sampling
apply to priority sampling as if the variable priority threshold was
fixed. As in the unweighted case, the result is obtained by a union
bound over threshold sampling events. In the unweighted case, we only
needed to consider the four threshold sampling error events (A), (B),
(A'), and (B'). However, now with weighted items, we are going to
reduce a priority sampling error event to the union of an unbounded
number of threshold sampling error events that happen with
geometrically decreasing probabilities.

\subsection{Notation and definitions}
Before formally presenting our priority sampling results, we introduce some notation and definitions.

\paragraph{Fractional subsets and inner products}
It is both convenient and natural to generalize our estimates from
regular subsets to {\em fractional subsets}, where for each $i\in I$,
there is a fraction $f_i\in[0,1]$ specifying that item $i$ contributes
$f_iw_i$ to the weight of fractional subset. A regular subset
corresponds to the special case where $f_i\in \{0,1\}$.

We are now interested in inner products between the fraction vector $f=(f_i)_{i\in I}$
and the vectors of weights or weight estimates. Our goal is to estimate
$fw=\sum_{i\in I} f_i w_i$.  With
threshold $t$, we estimate $fw$ as $f\widehat w^{\,t}=\sum_{i\in I} f_i
\widehat w^{\,t}_i=\sum_{i\in S} f_i \widehat w^{\,t}_i$.  With fixed
threshold $t$, we have $\E[\widehat w_i^{\,t}]=w_i$, so $\E[f_i\widehat
  w_i^{\,t}]=f_i w_i$ and $\E[f\widehat w^{\,t}]=fw$.

As an example, suppose we sampled grocery bills. For each bill sampled, we could check the fraction spent on candy, and based on that estimate the total amount spent on candy.

To emulate a standard subset $J$, we let $f$ be the characteristic
function of $J$, that is, $f_i=1$ if $i\in J$; otherwise $f_i=0$. In
fact, we will often identify a set with its characteristic vector,
so the weight of $J$ can be written as $Jw$ and estimated as $J\widehat
w^{\,t}$.

Using inner products
will simplify a lot of notation in our analysis.  The generalization
to fractional subsets comes for free in our analysis which is all
based on concentration bounds for sums of random variables
$X_i\in[0,1]$.

\paragraph{Notation for small and sampled weights}
With threshold $t$, we know that variability in the estimates is from items
$i$ with weight below $t$. We will generally use a subscript $_{<t}$
to denote the restriction to items $i$ with weights $w_i<t$,
e.g., $I_{<t}=\{i\in I | w_i<t\}$,
$w_{<t}=(w_i)_{i\in I_{<t}}$, and $fw_{<t}=\sum_{i\in I_{<t}} f_iw_i$. Notice that $fw_{<t}$
does not include $i$ with $w_i\geq t$ even if $f_iw_i<t$.

Above we defined $w_{<t}$ to denote the vector $(w_i)_{i\in I_{<t}}$ of
weights below $t$, and used it for the inner product
$fw_{<t}=\sum_{i\in I_{<t}} f_iw_i$. When it is clear from the context
that we need a number, not a vector, we will use $w_{<t}$ to denote
the sum of these weights, that is, $w_{<t}=\sum_{i\in I_{<t}}
w_i=\uw{<t}$ where $\ul{1}$ is the all 1s vector. Since $f_i\leq 1$ for all $i$, we always have $\fw{<t}\leq w_{<t}$.

We shall use subscript $_{\leq t}$, $_{\geq t}$, and $_{>t}$  
to denote the corresponding restriction to items with weight $\leq t$,
$\geq t$, and $>t$, respectively.

We also introduce a superscript $^t$ notation to denote the restriction to 
items sampled with threshold $t$, that is, items $i$ with $q_i>t$, so
$I^{\,t}_{<t}$ denotes the set of items with weights below $t$ that ended up
sampled. Identifying this set with its characteristic vector, we can
write our estimate with threshold $t$ as
\begin{equation}\label{eq:sample-vector}
f\widehat w^{\,t}=\fw{\geq t}+t(fI_{<t}^{\,t}).
\end{equation}

\paragraph{Error probability functions}
As mentioned earlier, we will reduce the priority sampling error event to the union
of an unbounded number of threshold sampling error events that happen
with geometrically decreasing probabilities. Our reduction will
hold for most hash functions, including 2-independent hash functions,
but to make such a claim clear, we have to carefully describe what
properties of the hash functions we rely on.

Assume that the threshold $t$ is fixed.
With reference to \req{eq:sample-vector}, the variability in
our estimate is all from 
\[fI_{<t}^{\,t}=\sum_{i\in I_{< t}} X_i
\textnormal{, where }X_i=f_i[h_i<w_i/t]\in [0,1].\]
As for the regular subsets in \req{eq:subsetvar}, let $X=\sum_{i\in I_{<t}} X_i$
and $\mu=\E[X]$. We are interested in an error probability
function $\wp$
such that for $\mu>0$, $\delta>0$, if $\mu=\E[X]$, then
\begin{equation}\label{eq:p}
\Pr[|X-\mu|>\delta\mu]\leq \wp(\mu,\delta).
\end{equation}
The error probability function $\wp$ that we can use depends on the quality of the hash function. For example, if the hash function is 2-independent, then 
$\Var[X]\leq\mu$, and then by Chebyshev's inequality, we can use
\begin{equation}\label{eq:chebyshev}
\wp^\textnormal{Chebyshev}(\mu,\delta)=1/(\delta^2\mu).
\end{equation}
In the case of full randomness, for $\delta\leq 1$, we could use
a standard 2-sided Chernoff bound (see, e.g., \cite{motwani95book})
\begin{equation}\label{eq:chernoff}
\wp^{\textnormal{Chernoff}_{\delta\leq 1}}(\mu,\delta)=2\exp(-\delta^2\mu/3).
\end{equation}
For most of our results, it is more natural to think of $\delta$ as
a function of $\mu$ and some target error probability $P\in (0,1)$, defining
 $\delta(\mu,P)$ such that
\begin{equation}\label{def:delta}
\mu(\mu,\delta(\mu,P))=P.
\end{equation}
Returning to threshold sampling with threshold $t$, by \req{eq:sample-vector} the error is $f\widehat w^{\,t}-fw=t(fI_{<t}^{\,t})-\fw{<t}$.
Hence
\begin{equation}\label{eq:threshold}
\Pr[|f\widehat w^{\,t}-fw|>\delta(\fw{<t}/t,P)\fw{<t}]\leq P.
\end{equation}
When we start analyzing priority sampling, we will need to relate the
probabilities of different threshold sampling events. This places some
constraints on the error probability function
$\wp$. Mathematically, it is convenient to allow $\wp$ to attain values
above $1$, but only values below $1$ are probabilistically interesting.
\begin{definition}\label{def:p} An error probability function 
$\wp:\mathbb R_{\geq 0}\times R_{\geq 0}\rightarrow R_{\geq 0}$ is {\em well-behaved\/} if
\begin{itemize}
\item[(a)] $\wp$ is continuous and strictly decreasing in both arguments. 
\item[(b)] If with the same absolute error we  decrease the expectancy, then
the probability goes down. Formally if $\mu'<\mu$ and 
$\mu'\delta'\geq \mu\delta$, then   $\wp(\mu',\delta')<\wp(\mu,\delta)$.
\end{itemize}
We also have an optional condition for cases where we only care for $\delta\leq 1$
as in \req{eq:chernoff}
\begin{itemize}
\item[(c)] If $\delta\leq 1$ and $\wp(\mu,\delta)<P_{\wp}^{(c)}$ for some
  constant $P_{\wp}^{(c)}$ depending on $\wp$, then $\wp(\mu,\delta)$ falls
  at least  proportionally to $\mu\delta^2$. Formally, if
  $\delta_0,\delta_1\leq 1$, $\wp(\mu_0,\delta_0)<1$, and
  $\mu_0\delta_0^2<\mu_1\delta_1^2$, then
\begin{equation}\label{eq:geometric}
\wp(\mu_0,\delta_0)\geq \frac{\mu_0\delta_0^2}{\mu_1\delta_1^2}\,
\wp(\mu_1,\delta_1).
\end{equation}
\end{itemize}
\end{definition}
We will use condition (c) to argue that
probabilities of different events fall geometrically. 
The condition is trivially satisfied with our 2-independent Chebyshev bound
\req{eq:chebyshev}, so we can just
set $P_{\wp^{\textnormal{Chebyshev}}}^{(c)}=1$.
The restrictions
in (c) are necessary for the Chernoff
bound \req{eq:chernoff}, 
$\wp^{\textnormal{Chernoff}_{\delta\leq 1}}=2\exp(-\delta^2\mu/3)$, which only falls fast enough for $\delta^2\mu/3\geq 1$, hence with
$\wp^{\textnormal{Chernoff}_{\delta\leq 1}}(\mu,\delta)\leq P_{\wp^{\textnormal{Chernoff}_{\delta\leq 1}}}^{(c)}=2/e$. As a further illustration, with
4-independence, we have the 
4th moment bound $\frac{\mu+3\mu^2}{(\delta\mu)^4}$
(see, e.g., \cite[Lemma 4.19]{KRS90}).
For $\mu\geq 1$, this is upper bounded by 
$\wp^{\textnormal{4th-moment}_{\mu\geq 1}}(\mu,\delta)=\left(\frac2{\mu\delta^2}\right)^2$. For
$\delta\leq 1$, the condition $\mu\geq 1$ is satisfied if
$\wp^{\textnormal{4th moment}_{\mu\geq 1}}(\mu,\delta)\leq 4$, so we can just
use $P_{\wp^{\textnormal{4th moment}}}^{(c)}=1$.

\paragraph{Threshold confidence intervals} In the case of threshold sampling with 
a fixed threshold $t$, we get some trivial confidence intervals
for the true value $fw$. 
The sample gives us the exact value $fw_{\geq t}$ for weights at least as  
big as $t$,
and an estimate $f\widehat w_{<t}^{\,t}$ for those below. 
Setting 
\begin{align*}
f\widehat w_{<t}^-&=\min\{x\,|\,(1+\delta(x/t,P))x\geq f\widehat w_{\geq t}^{\,t}\}\\
f\widehat w_{<t}^+&=\max\{x\,|\,(1-\delta(x/t,P))x\leq f\widehat w_{\geq t}^{\,t}\}
\end{align*}
we get 
\[\Pr\left[fw_{\geq t}+f\widehat w_{<t}^-\;\leq\; fw\;\leq \;fw_{\geq t}+f\widehat w_{<t}^+\right]\geq 1-2P.\]
We are going to show that similar bounds can be obtained for priority sampling.

\drop{
For a simpler to derive symmetric bound, 
if $\delta=\delta(f\widehat w_{<t}^{\,t}/t,P)\leq 1/\sqrt 2$,
then 
\[\Pr[fw_{<t}=(1\pm \sqrt 2\delta)f\widehat w_{<t}^{\,t}]< 1-2P.\]
This bound is tight in the upper bound. The point is
as follows. Let $\mu_1=2\mu_0$, $\delta_0=1/\sqrt 2$, and 
$\delta_1=1/2$. Then $\mu_1=(1+\sqrt 2 \delta_0)\mu_0$
and $\mu_0=(1-\delta_1)\mu_1$. Now $\mu_0\delta_0^2=\mu_1\delta_1^2$,
so by Definition \ref{def:p} (c), we have $\wp(\mu_0,\delta_0)=
\wp(\mu_1,\delta_1)$.}

\subsection{Priority sampling: the main result}\label{sec:main-results}
We are now ready to present our main technical result. We are
considering a random priority sample of size $k$, and let $\tau$ denote
the resulting priority threshold. The sample size $k$ and the target error probability $P$ are both fixed in advance of the random sampling.
\begin{theorem}\label{thm:main} Let the error probability function $\wp$ satisfy Definition \ref{def:p} including (c). 
With target error probability $P\leq P_\wp^{(c)}$, 
let
\[\delta=6\,\delta(f w_{<\tau}/(3\tau),P).\]
If $\delta\leq 2$, then 
\[\Pr[|f\widehat w^\tau-fw|>\delta fw_{<\tau}]\leq 6 P.\]
\end{theorem}
The above constants are not optimized, but with $O$-notation, some of
our statements would be less clear. Ignoring the constants
and the restriction $\delta \leq 2$, we see that our error bound for priority
sampling with threshold $\tau$ is of the same type as the one in \req{eq:threshold}
for threshold sampling with fixed threshold $t=\tau$.

The proof of Theorem \ref{thm:main} is rather convoluted. We 
consider a single priority sampling event with $k$ samples and priority
threshold $\tau$.  It assigns a random priority $q_i$ to each item, and this
defines a sample for any given threshold $t$. In particular, $\tau\geq
t$ if and only if we get at least $k+1$ samples with threshold $t$.
Note that $\Pr[\tau=t]=0$ for any given $t$.
We define 
\begin{eqnarray}
t_{\max}&=&\min\{\,t\;|\;\Pr[\tau\geq t]\leq P\,\}\label{eq:tmax-def}\\
t_{\min}&=&\max\{\,t\;|\;\Pr[\tau\!< t]\leq P\,\}\label{eq:tmin-def}
\end{eqnarray}
By definition, $\Pr[\tau\not\in [t_{\min},t_{\max})]\leq 2P$. By union, to prove Theorem \ref{thm:main}, it suffices to show that the following good event 
errors with probability at most $4P$:
\begin{align}
\forall t\in [t_{\min},t_{\max}): \delta&=\delta(f w_{<t}/(3 t),P)\leq 1/3\nonumber\\
&\implies|f\widehat w^{\,t}-fw|> 6\delta f w_{<t}\label{eq:good-t0}
\end{align}
Note that our variable $\delta$ is $6$ times smaller than the one
in the statement of Theorem \ref{thm:main}. This parameter change
will be more
convenient for the analysis.
What makes \req{eq:good-t0} very tricky to prove is that $\delta(f w_{<t}/(3 t),P)$
can vary a lot for different $t\in [t_{\min},t_{\max})$. 

\paragraph{Priority confidence intervals} The format of Theorem  \ref{thm:main}
makes it easy to derive confidence intervals like those for
threshold sampling. A priority sample with priority threshold $\tau$
gives us the exact value $fw_{\geq \tau}$ for weights at least as  
big as $\tau$,
and an estimate $f\widehat w_{<\tau}^\tau$ for those below. 
For an upper bound on $fw_{<\tau}$, we compute
\[f\widehat w_{<\tau}^+=\max\{x\,|\,\delta=6\,\delta(x/(3\tau),P))\wedge
(1-\delta)x\leq 
f\widehat w_{\geq \tau}^\tau\}.\]
Note that here in the upper bound, we only consider $\delta\leq 1$,
so we do not need to worry about the restriction $\delta< 2$ in Theorem \ref{thm:main}. For a lower bound on $fw_{<\tau}$, we use
\[f\widehat w_{<\tau}^-=\min\{x\,|\,\delta=6\,\delta(x/(3\tau),P))\leq 2\wedge
(1+\delta)x\geq f\widehat w_{\geq \tau}^\tau\}.\]
Here in the lower bound, the restriction $\delta=6\,\delta(x/(3\tau),P))\leq 2$
prevents us from deriving  a lower bound $x=f\widehat w_{<\tau}^-\leq 
f\widehat w_{\geq \tau}^\tau/3$. In such cases, we use the
trivial lower bound $x=f\widehat w_{<\tau}^-=0$ which in distance from
$f\widehat w_{<\tau}^\tau$ is at most 1.5 times bigger. Now, by Theorem \ref{thm:main},
\[\Pr\left[fw_{\geq \tau}+f\widehat w_{<\tau}^-\;\leq\; fw\;\leq \;fw_{\geq \tau}+f\widehat w_{<\tau}^+\right]\geq 1-12 P.\]
In cases where the exact part $fw_{\geq \tau}$ of an estimate is small compared with
the variable part $f\widehat w_{\geq \tau}^\tau$, we may be
interested in a non-zero lower bound 
$f\widehat w_{<\tau}^-$ even if it is smaller than $f\widehat w_{\geq \tau}^\tau/3$. To do this, we need bounds for larger $\delta$.

\paragraph{Large errors}
We are now going to present bounds that works for arbitrarily large relative
errors $\delta$. We assume a basic 
error probability function $\wp$ satisfying Definition \ref{def:p} (a) and (b)
while (c) may not be satisfied. The bounds we get are not as clean as those from 
Theorem \ref{thm:main}. In particular, they involve
$t_{\max}$ from \req{eq:tmax-def}. Since we are only worried about errors
$\delta>1$, we only have to
worry about positive errors.

\begin{theorem}\label{thm:tail} Set
\begin{align*}
\ell&=\lg\lceil (t_{\max}/\tau)\rceil\\
\delta&=\delta(f w_{<\tau}/\tau,\,P/\ell^2).
\end{align*}
Then 
\[\Pr[f\widehat w^\tau_{<\tau}> (2+2\delta)fw_{<\tau}] < 3P.\]
\end{theorem}
Complementing Theorem \ref{thm:main}, we only intend to use Theorem \ref{thm:tail} for large errors where $(2+2\delta)=O(\delta)$.  We wish
to provide a probabilistic lower bound 
for $fw_{<\tau}$. Unfortunately, we do not know $t_{\max}$ which
depends on the whole weight vector $(w_i)_{i\in I}$. However, 
based our priority sample, it is not
hard to generate a probabilistic upper bound $\ol t_{\max}$ on $t_{\max}$ such
that $\Pr[\ol t_{\max}<t_{\max}]\leq P$. We set
\begin{equation}\label{eq:oll}
\ol \ell=\lceil \lg (\ol t_{\max}/\tau)\rceil
\end{equation}
\begin{equation}\label{eq:conf-tail}
f\widehat w_{<\tau}^-=\min\{x\,|\,\ol\delta=\delta(x/\tau,P/\ol\ell^{\,2})\wedge
2(1+\ol\delta)x\geq f\widehat w_{< \tau}^\tau\}.
\end{equation}
Then by Theorem \ref{thm:tail},
\[\Pr\left[fw_{<\tau}\;\geq \;f\widehat w_{<\tau}^-\right]\geq 1-4P.\]
To see this, let $f\widehat w_{<\tau}^*$ be the value we would have
obtained if we had computed $f\widehat w_{<\tau}^-$ using the
real $t_{\max}$. Our error event is that 
$\ol{t}_{\max}<t_{\max}$ or $fw_{<\tau}\;< \;f\widehat w_{<\tau}^*$.
The former happens with probability at most $P$, and 
Theorem \ref{thm:tail} states that the latter happens with probability at most $3P$.
Hence none of these error events happen with probability at least $1-4P$,
but then $\ol{t}_{\max}\geq t_{\max}$, implying $f\widehat w_{<\tau}^-\leq
f\widehat w_{<\tau}^*\leq fw_{<\tau}$.

Theorem \ref{thm:main}, our main result, is proved in Sections
\ref{sec:tau}--\ref{sec:layers}. Theorem \ref{thm:tail}, which is much
easier, is proved in Section \ref{sec:large-error}.  In Section
\ref{sec:upper-upper} we show how to compute the $\ol{t}_{\max}$ used
for confidence lower bound with Theorem \ref{thm:tail}. Finally, in Section \ref{sec:heavy-tail} we will argue that we for typical weight distributions 
expect to get $\ol\ell=1$.

\subsection{The priority threshold}\label{sec:tau}
To prove Theorem \ref{thm:main} we need a handle on 
the variable priority threshold. With priority sampling we specify the number $k$ of samples, and use
as threshold $\tau$ the $(k+1)$th priority. Recall for any threshold
$t$ that the subscript $_{\leq t}$ indicates restriction to items with
weight below $t$. To relate this notation to a priority sample of some
specified size $k$, we let $k_{\leq t}$ denote $k$ minus the number of
items with weight bigger than $t$.  We define $k_{< t}$ accordingly.
With threshold $t$, the expected
number of samples is $E[|I^{\,t}|]=k-k_{\leq t}+w_{\leq t}/t=k-k_{< t}+w_{< t}/t$.
The last equality is because weights $w_i=t$ cancel out.

\paragraph{The ideal threshold}
We define the {\em
  ideal threshold $t^*$} to be the one leading to an expected number
of exactly $k$ samples, that is
$w_{\leq t^*}=t^* k_{\leq t^*}$.
\begin{lemma}\label{lem:mean-t} 
$t k_{\leq t}-w_{\leq t}$ is strictly increasing in $t$, so
(a) $t^*$ is uniquely defined, (b) $w_{\leq t}>k_{\leq t}t$ for any $t< t^*$,
and (c) $w_{\leq t}< k_{\leq t}t$ for any $t> t^*$. 
\end{lemma}
\begin{proof}
If $t$ increases without passing any weight value, then
$k_{\leq t}$ and $w_{\leq t}$ are not changed, and the
statement is trivial. When $t$ reaches the value of
some weight $w_i$, then both $t k_{\leq t}$ and $w_{\leq t}$
are increased by the same value $w_i$ (if there are $j$ weights with
value $w_i$, the increase is by $jw_i$).
\end{proof}
We would like to claim that the priority sampling threshold $\tau$ is
concentrated around $t^*$, but this may be far from true.  To
illustrate what makes things tricky to analyze, consider the case
where, say, we have $k-1$ weights of size $t^*$, and then a lot of
small weights that sum to $t^*$. In this case we get an upper bound on
$\tau$ which is close to $t^*$, but we do not get any good lower bound
on $\tau$ even if we have full randomness. On the other hand, in this
case, it is only little weight that is affected by the downwards
variance in $\tau$.

\subsection{Tightening the gap}
The following lemmas give us a much tighter understanding
of $t\in [t_{\min},t_{\max})$.
\begin{lemma}\label{lem:large-t}  
If $t\geq t^*$ then 
$t(1+\delta(w_{\leq t}/t_{\max},P))\geq t_{\max}$. 
\end{lemma}
\begin{proof}
Let $\delta=\delta(w_{\leq t}/t_{\max},P)$ and $T=(1+\delta)t$. Assume for a contradiction that $T<t_{\max}$.
By Lemma~\ref{lem:mean-t}~(c), since $t\geq t^*$, we have $w_{\leq t}\leq 
tk_{\leq t}$, so 
\[k_{\leq t}\geq w_{\leq t}/t=(1+\delta)w_{\leq t}/T\]
For the priority threshold to be as big as $T$ we need
$|I^T_{\leq t}|\geq k_{\leq t}+1$ and $\E[|I^T_{\leq t}|]=w_{\leq t}/T$,
so 
\[\Pr[\tau\geq T]\leq \wp(w_{\leq t}/T,\delta)\leq \wp(w_{\leq t}/t_{\max},\delta)=P.\]
But this contradicts the minimality of $t_{\max}$ from \req{eq:tmax-def}.
\end{proof}

\begin{lemma}\label{lem:small-t}  
If $t\in[t_{\min},t^*)$ then $t\geq (1-\delta(w_{\leq t}/t,P))t^*$. 
\end{lemma}
\begin{proof}
Let $\delta=\delta(w_{\leq t}/t,P)$.
The proof is by contradiction against our maximal lower bound $t_{\min}$ from \req{eq:tmin-def}.
The priority threshold is smaller than $t$ if 
$|I_{\leq t}^{t}|\leq k_{\leq t}$. The
expectancy is $\E[|I_{\leq t}^{t}|]=w_{\leq  t}/t$.
Suppose for a contradiction that $(1-\delta)w_{\leq t}/t> k_{\leq t}$.
Then
\[\Pr[\tau<t]< \wp(w_{\leq t}/t,\delta)=P\textnormal,\]
implying that $t\geq t_{\min}$ is a lower bound. If $t>t_{\min}$ this
contradicts the maximality of $t_{\min}$. Otherwise, we pick an 
infinitesimally larger $t'>t$ with no weights in $(t,t']$, that is,
$w_{\leq t'}=w_{\leq t}$ and $k_{\leq t'}=k_{\leq t}$. By Definition \ref{def:p} (a),
we get a corresponding infinitesimally larger 
$\delta'=\delta(w_{\leq t'}/t',P)>\delta$, and then we still have
$(1-\delta')w_{\leq t'}/t'>k_{\leq t'}$, implying that
$t'>t=t_{\min}$ is a lower bound contradicting the maximality of $t_{\min}$.
Thus
we conclude that $(1-\delta)w_{\leq  t}/t\leq  k_{\leq  t}$, or equivalently, 
$t\geq (1-\delta)w_{\leq  t}/k_{\leq  t}$.
Finally by Lemma \ref{lem:mean-t} (b), since $t<t^*$, we have $w_{\leq  t}/k_{\leq t}
\geq t^*$. This completes the proof that $t\geq (1-\delta_\ell)t^*$.
\end{proof}
In order to give a joint analysis for $t$ bigger and smaller than $t^*$,
we make a conservative combination of Lemma \ref{lem:large-t} and
\ref{lem:small-t}.
\begin{lemma}\label{lem:all-t}  
Suppose $[t^-,t^+)=[t_{\min},t^*)$ or $[t^-,t^+)=[t^*,t_{\max})$.
If $t\in [t^-,t^+)$ then
$t\geq (1-\delta(w_{\leq t}/t^+,P))t^+$. 
\end{lemma}
\begin{proof} Let $\delta=\delta(w_{\leq t}/t^+,P)$.
If $[t^-,t^+)=[t^*,t_{\max})$, 
by Lemma \ref{lem:large-t}, 
$t\geq t^+/(1+\delta)\geq (1-\delta)t^+$.
If $[t^-,t^+)=[t_{\min},t^*)$, 
by Lemma \ref{lem:small-t}, 
$t\geq (1-\delta(w_{\leq t}/t,P))t^+>(1-\delta)t^+$.
The last inequality uses that $\delta(w_{\leq t}/t,P)\leq \delta$.
This is because $w_{\leq t}/t\geq w_{\leq t}/t^+$, so by
Definition \ref{def:p} (a), $\wp(w_{\leq t}/t,\delta)\geq 
\wp(w_{\leq t}/t^+,\delta)=P$.
\end{proof}
Loosing a factor 2 in the error probability to cover $t$ bigger or smaller
than $t^*$, our good event \req{eq:good-t0}
reduces to
\begin{align}
\forall t\in [t^-,t^+): \delta&=\delta(f w_{<t}/(3 t),P)\leq 1/3\nonumber\\
&\implies|f\widehat w^{\,t}-fw|> 6\delta f w_{<t}\label{eq:good-t1}
\end{align}
By Lemma \ref{lem:all-t}, the bound on $\delta$ implies that
$t\geq (2/3) t^+$. Therefore \req{eq:good-t1} is implied by
\begin{align}
\forall t\in [t^-,t^+): \delta&=\delta(f w_{<t}/(2 t^+),P)\leq 1/3\nonumber\\
&\implies|f\widehat w^{\,t}-fw|\leq 6\delta f w_{<t}.\label{eq:good-t2}
\end{align}
One advantage of dealing with $f\widehat w_{<t}/t^+$ instead
of $f\widehat w_{<t}/t$ is that $f\widehat w_{<t}/t^+$ is
proportional to $f\widehat w_{<t}$ hence increasing in $t$ whereas
$f\widehat w_{<t}/t$ may not be monotone in $t$. Then $\delta(f w_{<{t'}^-}/(2 t^+),P)$ is decreasing in $t$. 
If 
$\delta(f w_{<t^-}/(2 t^+),P)> 1/3$, we let ${t'}^-$ be the
smallest value such that $\delta(f w_{<{t'}^-}/(2 t^+),P)\leq 1/3$;
otherwise we set ${t'}^-=t^-$. 
Then \req{eq:good-t2} is equivalent to
\begin{align}
\forall t\in [{t'}^-,t^+): \delta&=\delta(f w_{<t}/(2 t^+),P)\nonumber\\
&\implies|f\widehat w^{\,t}-fw|\leq 6\delta f w_{<t}.\label{eq:good-t3}
\end{align}

\subsection{Dividing into layers}\label{sec:layers}
We now define a sequence $t_0>t_1>\cdots>t_{L+1}$ of decreasing
thresholds with $t_0=t^+$ and $t_{L+1}={t'}^-$.
For $\ell=0,..,L$ we require
\begin{equation}\label{eq:heavy-rest}
fw_{\leq t_{\ell+1}}\geq \fw{<t_{\ell}}/2.
\end{equation}
For $\ell<L$, we pick $t_{\ell+1}$ smallest possible
satisfying \req{eq:heavy-rest}. Then
\begin{equation}\label{eq:halving}
fw_{<t_{\ell+1}}<\fw{<t_{\ell}}/2.
\end{equation}
We arrive $\ell=L$ when $fw_{\leq {t'}^-}\geq \fw{<t_{\ell}}/2$, and then
we set $t_{L+1}={t'}^-$. 

For each ``layer'' $\ell\leq L$, we define
\begin{equation}\label{eq:dl}
\delta_{\ell}=\delta(fw_{<t_{\ell}}/(2t^+),P),
\end{equation}
noting that this is the same value as we would use for $t=t_\ell$ in
\req{eq:good-t3}. By definition, for all $\ell\leq L$, we have
$\wp(fw_{<t_{\ell}}/(2t^+),\delta_\ell)=P$.
Since $t_\ell>{t'}^-$, we have $\delta_\ell\leq 1/3\leq 1$, so
it follows from Definition \ref{def:p} (c) that there
is a constant $C$ such that $fw_{<t_\ell}\delta_\ell^2=C$ for all $\ell\leq L$.
For $\ell=1,...,L$, by \req{eq:halving}, we have
$fw_{< t_{\ell}}<fw_{<t_{\ell-1}}/2$. Therefore
\begin{align}
\delta_{\ell}&>\sqrt 2\,\delta_{\ell-1}\label{eq:delta-up}\\
\delta_\ell fw_{<t_\ell}&<\delta_{\ell-1} fw_{<t_{\ell-1}}/\sqrt 2\label{eq:error-down}
\end{align}
This will correspond to an effect where the relative errors are
geometrically increasing while the absolute errors are geometrically
decreasing. Another important thing to notice is
that by \req{eq:heavy-rest}, $fw_{\leq t_{\ell+1}}\geq \fw{<t_{\ell}}/2$,
so $\delta(fw_{<t_{\ell+1}}/t^+,P)\leq
\delta(fw_{<t_{\ell}}/(2t^+),P)=\delta_{\ell}$. Therefore, by Lemma \ref{lem:all-t},
\begin{equation}\label{eq:tl+1}
t_{\ell+1}\geq (1-\delta_\ell)t^+.
\end{equation}

\paragraph{Good layers}
For each layer $\ell<L$, our good event will be that for weights in
$[t_{\ell+1},t_\ell)$, the relative estimate error is bounded by
  $2\delta_{\ell}$. Formally
\begin{align}
\forall t\in[t_{\ell+1},t^+):&\left|f\widehat w_{[t_{\ell+1},t_\ell)}^{\,t}-
f w_{[t_{\ell+1},t_\ell)}\right|\nonumber\\ &\quad\leq
2\delta_\ell f w_{[t_{\ell+1},t_\ell)}\label{eq:good-l}
\end{align}
Above, the subscript $_{[t_{\ell+1},t_\ell)}$ denotes the restriction to
items $i$ with weights $w_i\in [t_{\ell+1},t_\ell)$. The last
layer $L$ is special in that we want to consider all weights
below $t_L$. Here the good event is
that 
\begin{align}
\forall t\in[t_{L+1},t^+):&\left|f\widehat w_{<t_L}^{\,t}-
f w_{<t_L}\right|\nonumber\\ &\quad\leq
3\delta_L f w_{< t_L}\label{eq:good-L}
\end{align}
To prove Theorem \ref{thm:main}, we are going to prove two statements.
\begin{itemize}
\item Assume that all layers are good satisfying  \req{eq:good-l} and \req{eq:good-L}. If for any $t\in({t'}^-,t^+]$ we add up the errors from all relevant
layers, then the total error is bounded by $6\delta(f w_{<t}/(2 t^+),P) f w_{<t}$, so \req{eq:good-t3} satisfied.
\item If $P_\ell$ is the probability that a layer $\ell$ fails, then
the $P_\ell$ are geometrically increasing and $P_L=O(P)$, so by union,
the probability that any layer fails is $O(P)$.
\end{itemize}

\paragraph{Adding layer errors}
Assuming that all layers are good satisfying \req{eq:good-l} and
\req{eq:good-L}, we pick an arbitrary threshold $t\in [{t'}^-,t^+)$.
We which to bound the estimate error $|f\widehat w^{\,t}-fw|$. 

Let $h$ be the layer such that $t\in [t_{h-1},t_h)$. We can
only have estimate errors from weights below $t< t_h$, so 
\begin{align*}
|f\widehat w^{\,t}-fw|&\leq\sum_{\ell=h}^{L-1} 
\left|f\widehat w_{[t_{\ell+1},t_\ell)}^{\,t}-
f w_{[t_{\ell+1},t_\ell)}\right|+\left|f\widehat w_{<t_L}^{\,t}-
f w_{<t_L}\right|\\
&\leq\sum_{\ell=h}^{L-1} 
2\delta_\ell w_{[t_{\ell+1},t_\ell)}+3\delta_L f w_{< t_L}\\
&=\sum_{\ell=h}^{L-1} 
2\delta_\ell (f w_{<t_{\ell}}-f w_{<t_{\ell+1}}) +3\delta_L f w_{< t_L}
\end{align*}
By \req{eq:delta-up}, the $\delta_\ell$ are increasing, so in the
above sum, every $f w_{<t_{\ell}}$ appears with a positive coefficient.
It follows that we could only get a larger sum if \req{eq:halving} was
more than tight with $fw_{<t_\ell}=fw_{<t_{\ell-1}}/2$ for $\ell\leq L$.
Then we would have $f w_{<t_{\ell}}-f w_{<t_{\ell+1}}=f w_{<t_{\ell}}/2$
and corresponding to \req{eq:error-down}, $\delta_\ell fw_{<t_\ell}=
\delta_{\ell-1}fw_{<t_{\ell-1}}/\sqrt 2$. 
Thus we get 
\begin{align*}
|f\widehat w^{\,t}-fw|&\leq\sum_{\ell=h}^{L-1} 
2\delta_\ell (f w_{<t_{\ell}}-f w_{<t_{\ell+1}})+3\delta_L f w_{< t_L}\\
&\leq\sum_{\ell=h}^{L-1} 
\delta_h f w_{<t_{h}}/{\sqrt2}^{\ell-h} +3\delta_h f w_{<t_{h}}/{\sqrt2}^{L-h}\\
&< 3\delta_h f w_{<t_{h}}
\end{align*}
The last inequality exploits that $\sum_{i=1}^{\infty}1/\sqrt 2=1/(1-1/\sqrt 2)<3$.
Since $t\leq t_h$, we have
$\delta(fw_{<t}/(2t^+),P)\geq\delta(fw_{<{t_h}}/(2t^+),P)=\delta_h$.
Also, $t>t_{h+1}$, so by \req{eq:heavy-rest}, 
$fw_{< t}\geq fw_{\leq t_{h+1}}\geq fw_{< t_h}/2$, so
\[\delta(fw_{<t}/(2t^+),P)\,fw_{< t}\geq \delta_{h} fw_{<t_{h}}/2.\]
Therefore
\[|f\widehat w^{\,t}-fw|\leq 3\delta_h f w_{<t_{h}}\leq
6\,\delta(fw_{<t}/(2t^+),P)\,fw_{< t}.\]
Thus we conclude that \req{eq:good-t3} follows
from \req{eq:good-l} and
\req{eq:good-L}.

\paragraph{Intermediate layer error probabilities}
We now consider the intermediate layers $\ell=0,...,L-1$. 
We want to show that the probability $P_\ell$ of violating
\req{eq:good-l} increases geometrically with $\ell$, yet remains
bounded by $P$. 
First we consider the upper bound part of \req{eq:good-l} 
\begin{equation}\label{eq:upper-l}
\forall t\in[t_{\ell+1},t^+):f\widehat w_{[t_{\ell+1},t_\ell)}^{\,t}\leq
(1+2\delta_\ell)f w_{[t_{\ell+1},t_\ell)}.
\end{equation}
We claim that it can never be violated. The worst that
can happen is that every item $i$ in the layer gets sampled, and
the estimate is at most $f_i t^+$. However, the items all have weight at
least $t_{\ell+1}$ and by \req{eq:tl+1}, $t_{\ell+1}\geq (1-\delta_\ell)t^+$.
The increase is thus by at most a factor $t^+/t_{\ell+1}\leq 
1/(1-\delta_{\ell})$, which for $\delta_\ell\leq 1/3$ is at most $(1+2\delta_\ell)$. Thus \req{eq:upper-l}
is satisfied regardless of the random choices.

We now consider the lower bound part of \req{eq:good-l} 
\begin{equation}\label{eq:lower-l}
\forall t\in [t_{\ell+1},t^+):f\widehat w_{[t_{\ell+1},t_\ell)}^{\,t}\geq
(1-2\delta_\ell)f w_{[t_{\ell+1},t_\ell)}.
\end{equation}
This event could happen. To bound the probability, we will focus on
the loss $f w_{[t_{\ell+1},t_\ell)}-f\widehat
w_{[t_{\ell+1},t_\ell)}^{\,t}$. When bounding the loss, we do not
consider the gain from possible overestimates of sampled
items. We only consider the actual losses $f_iw_i$ from unsampled items
$i$. Conservatively, we consider an item $i$ lost if $q_i\leq t^+$. This
includes any item unsampled with some threshold $t\leq t^+$. The loss
for every threshold $t\leq t^+$ is thus bounded as
\[f w_{[t_{\ell+1},t_\ell)}-f\widehat
  w_{[t_{\ell+1},t_\ell)}^{\,t}\leq \sum_{i:w_i\in [t_{\ell+1},t_\ell)}
  [q_i\leq t^+]f_iw_i.\]
We know that $w_i\geq t_{\ell+1}\geq (1-\delta_\ell)t^+$. Therefore
\begin{align*}
\Pr[q_i\leq t^+]&=\Pr[h_i\geq w_i/t^+]\leq \Pr[h_i\geq t_{\ell+1}/t^+]\\
&= 1-t_{\ell+1}/t^+\leq \delta_\ell.
\end{align*}
The expected loss from layer $\ell$ is thus bounded by
\[\sum_{i:w_i\in[t_{\ell+1},t_\ell)}\delta_\ell f_iw_i=\delta_{\ell}
fw_{[t_{\ell+1},t_\ell)}.\] 
For \req{eq:lower-l} to fail, we need a loss that is twice this big, that is,
$2\delta_{\ell}fw_{[t_{\ell+1},t_\ell)}$. We know that items $i$'s loss contribution
$[q_i\leq t^+]f_iw_i$ depends only on $h_i$ and that
it is at most $t^+$. The probability of violating
\req{eq:lower-l}  is therefore bounded by
\[\wp(\delta_{\ell}fw_{[t_{\ell+1},t_\ell)}/t^+,\,1)\leq
\wp(\delta_{\ell}fw_{<t_\ell}/(2t^+),\,1).\] 
Let $\mu_\ell=\delta_{\ell}fw_{<t_\ell}/(2t^+)$. Then 
$P_\ell=\wp(\mu_\ell,1)$ is our bound on the error probability.
From \req{eq:error-down}, we know that
$\delta_{\ell}fw_{<t_\ell}<\delta_{\ell-1}fw_{<t_{\ell-1}}/\sqrt 2$,
so $\mu_\ell<\mu_{\ell-1}/\sqrt 2$. It follows from 
Definition \ref{def:p} (c) that $\wp(\mu_\ell,1)>\sqrt 2\,\wp(\mu_{\ell-1},1)$,
so the $P_\ell$ are geometrically increasing, and their sum is bounded by
$3P_{L-1}$. 

Finally, by definition \req{eq:dl},
$\wp(fw_{<t_{L-1}}/(2t^+),\,\delta_{L-1})=P$. To compare $P$ with
$P_{L-1}$, by Definition \ref{def:p} (c), we compare
$fw_{<t_{L-1}}/(2t^+)\delta_{L-1}^2$ with $\mu_{L-1}1^2=
\delta_{L-1}fw_{<t_{L-1}}/(2t^+)$, and conclude that $P_{L-1}\leq
\delta_{L-1}P\leq P/3$. The probability that any intermediate layer
$\ell<L$ fails \req{eq:lower-l}  is thus at most
$P$. Since \req{eq:lower-l} was always satisfied, we conclude that
\req{eq:good-L} is satisfied for all layers $\ell<L$ with probability
$P$.

\paragraph{The last layer}
We now consider items $i$ with  weights below $t_L$. On the upper bound
side, the  good event \req{eq:good-L} states that
\begin{equation}\label{eq:upper-L}
\forall t\in[t_{L+1},t^+):f\widehat w_{<t_L}^{\,t}\leq (1+3\delta_L) f w_{< t_L}.
\end{equation}
We will show that \req{eq:upper-L} fails with probability less than $P$. For an upper bound on the estimate with any threshold 
$t\in [t_{L+1},t^+)$, we include
item $i$ if $q_i>t_{L+1}$, and if so, we give it at an estimate of $f_i t^+$
which is bigger than the sampled estimate with threshold $t\leq t^+$.
The result is  at most a factor 
$t^+/t_{L+1}$ bigger than in the sampled estimate with threshold $t_{L+1}$, and
by \req{eq:tl+1}, $t_{L+1}\geq (1-\delta_L)t^+$. Thus, regardless of the
random choices made,  we conclude that
\[\forall t\in[t_{L+1},t^+):f\widehat w_{<t_L}^{\,t}
\leq f\widehat w_{\leq t_L}^{t_{L+1}}/(1-\delta_L).\]
Consider the following error event:
\begin{equation}\label{eq:upper-L-error}
f\widehat w_{\leq t_L}^{t_{L+1}}>(1+\delta_L)fw_{\leq t_L}
\end{equation}
The maximal item contribution to $f\widehat w_{\leq t_L}^{t_{L+1}}$
is bounded by $t_L\leq t^+$, so the probability of \req{eq:upper-L-error}
is bounded by 
\[\wp(fw_{\leq t_L}/t^+,\delta_L)\leq \wp(fw_{\leq t_L}/(2t^+),\delta_L)/2=P/2.\]
If \req{eq:upper-L-error} does not happen, then since $\delta_L\leq 1/3$,
\begin{align*}
\forall t\in[t_{L+1},t^+):f\widehat w_{<t_L}^{\,t}&\leq(1+\delta_L)fw_{\leq t_L}/(1-\delta_L)\\
&\leq (1+3\delta_L)fw_{\leq t_L},
\end{align*}
which is the statement of \req{eq:upper-L}. We conclude that
\req{eq:upper-L} fails with probability at most $P/2$.

We now consider the 
lower bound side of \req{eq:good-L} which states that
\begin{equation}\label{eq:lower-L}
\forall t\in[t_{L+1},t^+):f\widehat w_{<t_L}^{\,t}\geq (1-3\delta_L) f w_{< t_L}
\end{equation}
The analysis is very symmetric to the upper bound case.
For a lower bound for weights $w_i<t_{L}$ with any threshold
$t\in[t_{L+1},t^+)$, we only include $i$ if $q_i>t^+$, and if so, we
only give $i$ the estimate $f_it_{L+1}$ which is smaller than
the sampled estimate with threshold $t>t_{L+1}$. Our
samples are exactly the same as those we would get with
threshold $t^+$, and our estimates are smaller by a factor $t_{L+1}/t^+\geq (1-\delta_L)$,
so we conclude that regardless of the random choices,
\[\forall t\in[t_{L+1},t^+):f\widehat w_{<t_L}^{\,t}\geq f\widehat 
w_{<t_L}^{t^+}(1-\delta_L)\]
We now consider the error event:
\begin{equation}\label{eq:lower-L-error}
f\widehat w_{\leq t_L}^{t^+}<(1-\delta_L)fw_{\leq t_L}
\end{equation}
The maximal item contribution to $f\widehat w_{\leq t_L}^{t^+}$ is $t^+$,
so as for \req{eq:upper-L-error}, we get that the probability
of \req{eq:lower-L-error} is bounded by $\wp(fw_{\leq t_L}/t^+,\delta_L)\leq P/2$.

If \req{eq:lower-L-error} does not happen, then since $\delta_L\leq 1/3$,
\begin{align*}
\forall t\in[t_{L+1},t^+):f\widehat w_{<t_L}^{\,t}&\geq(1-\delta_L)fw_{\leq t_L}(1-\delta_L)\\
&>(1-2\delta_L)fw_{\leq t_L},
\end{align*}
which implies \req{eq:lower-L}. We conclude that
\req{eq:upper-L} fails with probability at most $P/2$. Including the
probability of an upper bound error \req{eq:upper-L-error}, we get that \req{eq:good-L} fails
with probability at most $P$. 

\paragraph{Summing up}
Above we proved that the probability that \req{eq:good-l} failed for any
layer $\ell<L$ was at most $P$. We also saw that \req{eq:good-L} failed
with probability P. If none of them fail, we proved
that \req{eq:good-t3} and hence \req{eq:good-t1} was satisfied,
so \req{eq:good-t1} fails with probability at most $2P$.
However, for \req{eq:good-t0} we need \req{eq:good-t1} both
for $[t^-,t^+)=[t_{\min},t^*)$ and for $[t^-,t^+)=[t^*,t_{\max})$, so
\req{eq:good-t0} fails with probability at most $4P$. Finally, we
need to consider both the case that $\tau>t_{\max}$ and $\tau\leq t_{\min}$.
Either of these events happens with probability at most $P$, so we
conclude that the overall error probability is at most $6P$.
If no error happened and $\delta=\delta(f w_{<t}/(3 t),P)\leq 1/3$,
then $|f\widehat w^{\,t}-fw|> 6\delta f w_{<t}$. This completes
the proof of Theorem \ref{thm:main}.

\subsection{Large errors}\label{sec:large-error}
The limitation of Theorem \ref{thm:main} is that it can only be used
to bound the probability that the estimate error $|f\widehat w^\tau-fw|$
is bigger than $2f w_{<\tau}$. Note that errors above $f
w_{<\tau}$ can only be overestimates. We will now target larger errors
and prove the statement of Theorem \ref{thm:tail}:
\begin{quote}{\em Set
\begin{align*}
\ell&=\lceil \lg (t_{\max}/\tau)\rceil\\
\delta&=\delta(f w_{<\tau}/\tau,\,P/\ell^2).
\end{align*}
Then 
\[\Pr[f\widehat w^\tau_{<\tau}> (2+2\delta)fw_{<\tau}]< 3P\]
}\end{quote}
\begin{proof}[ of Theorem \ref{thm:tail}]
Since we are
targeting arbitrarily large relative errors $\delta$, for the
probability function $\wp$, we can only assume conditions (a) and (b) in
Definition \ref{def:p}.

We will use some of the same ideas as we used for the last layer in Section \ref{sec:layers}, but 
tuned for our situation. 
We will study intervals based on $t_\ell=t_{\max}/2^\ell$ for $\ell=1,2,...$. Interval
$\ell$ is for thresholds $t\in [t_\ell,t_{\ell-1})=[t_\ell,2t_{\ell})$, so $t<t_{\max}$ belongs to interval $\ell=\lceil \lg (t_{\max}/t)\rceil$.
To define the error for interval $\ell$, set
\[\delta_\ell=\wp(fw_{< t_{\ell}}/t_{\ell},\, P/\ell^2).\]
The good non-error event for interval $\ell$ is that 
\begin{equation}\label{eq:no-large-error}
f\widehat w_{< t_{\ell}}^{t_{\ell}}\leq (1+\delta_\ell)fw_{< t_{\ell}}.
\end{equation}
By definition, the probability that \req{eq:no-large-error} is violated is 
at most $P/\ell^2$, so the probability of failure for any $\ell$ is bounded by $\sum_{\ell=1}^\infty P/\ell^2\leq 
P\pi^2/6<1.65 P$. The probability that $\tau=t_{\max}$ is zero, so by
\req{eq:tmax-def}, the event
\begin{equation}\label{eq:tmax}
\tau< t_{\max}.
\end{equation}
is violated with probability at most $P$. Our total error probability
is thus bounded by $2.65 P<3P$. Below we assume no errors, that is,
\req{eq:no-large-error} holds for all $\ell$ and so does 
\req{eq:tmax}.

Consider an arbitrary threshold $t\in (0,t_{\max})$ and let $\ell$ be
such that $t\in[t_\ell,2t_{\ell})$. 
We can only have errors for weights $w_i<t$, so we want an upper
bound on $f\widehat w_{<t}^{\,t}$. The basic idea for an upper bound
is to say that we sample all items with priority above $t_{\ell}$, just
as in the estimate $f\widehat w_{<t_{\ell}}^{t_{\ell}}$, but instead of giving
sampled item $i$ estimate $\max\{w_i,t_{\ell}\}$, it gets value 
$\max\{w_i,t\}$ which
is at most $t/t_{\ell}< 2$ times bigger. Thus, regardless of the random choices, 
\[f\widehat w^{t}_{<t}<2 f\widehat w^{t_{\ell}}_{<t}.\]
Assuming no error as in \req{eq:no-large-error}, we get
\[f\widehat w^{t_{\ell}}_{<t}=
f w_{[t_{\ell},t]}+f\widehat w^{t_{\ell}}_{<t_{\ell}}\leq 
f w_{<t}+\delta_{\ell}f w_{<t_{\ell}}.\]
For the further analysis, we need a general lemma.
\begin{lemma}\label{lem:decr-t} For thresholds $t',t$,  and relative errors
$\delta',\delta$, 
if $t'<t$ and $\wp(fw_{< t'}/t',\delta_{< t'})=\wp(fw_{\leq t}/t,\delta_{\leq t})$,
then 
\[\delta' f w_{<t'}<\delta f w_{<t}.\]
Hence, for any fixed target error probability $Q$ in \req{eq:threshold}, the
target error 
\[\delta(\fw{<t}/t,Q)\fw{<t}\]
decreases together with the threshold $t$.
\end{lemma}
\begin{proof} We will divide the decrease from $t$ to $t'$ into 
a series of atomic decreases.
The first atomic ``decrease'' 
is from $fw_{\leq t}$ to $fw_{<t}$. 
This makes no
difference unless there are weights equal to $t$ so
that $fw_{<t}<fw_{\leq t}$. Assume this is the case and suppose 
$\wp(fw_{< t}/t,\delta_{< t})=\wp(fw_{\leq t}/t,\delta_{\leq t})$. 
Since $fw_{\leq t}/t<fw_{<t}/t$, it
follows directly from 
(b) that 
$\delta_{<t}fw_{<t}/t<\delta_{\leq t}w_{\leq t}/t$,
hence that $\delta_{<t}fw_{<t}<\delta_{\leq t}w_{\leq t}$.

The other atomic decrease we consider is from $fw_{<t}$ to $fw_{\leq t'}$
where $t'<t$ and with no weights in $(t',t)$, hence with 
$fw_{\leq t'}=fw_{<t}$. Suppose 
$\wp(fw_{\leq t'}/t',\delta_{\leq t'})=\wp(fw_{<t}/t,\delta_{< t})$.
Since $t'<t$, $fw_{\leq t'}/t'>fw_{<t}/t$, so by 
(a),  $\delta_{\leq t'}<\delta_{<t'}$.
It follows that $\delta_{\leq t'}fw_{\leq t'}<\delta_{<t}fw_{<t}$. Alternating between these two atomic decreases, we can implement an arbitrary decrease
in the threshold as required for the lemma.
\end{proof}
Let 
\[\delta=\delta(w_{<t}/t,P/\ell^2)\]
By Lemma \ref{lem:decr-t}, 
since $t>t_{\ell}$, we have
$\delta_{\ell}f w_{<t_{\ell}}<\delta f w_{<t}$, so
\[
f\widehat w^{t_{\ell}}_{<t}\leq f w_{<t}+\delta_{\ell}fw_{<t_{\ell}}
< f w_{<t}+\delta f w_{<t_{\ell}}.\]
We thus conclude
\begin{equation}\label{eq:large-error}
\forall t\in (0,t_{\max}): f\widehat w^{\,t}_{<t}<2f\widehat w^{t_{\ell}}_{<t}<2(1+\delta)f w_{<t}\textnormal,
\end{equation}
This completes the proof of Theorem \ref{thm:tail}.
\end{proof}

\subsection{Upper bounding the upper bound}\label{sec:upper-upper}
Theorem \ref{thm:tail} uses the threshold upper bound $t_{\max}$ which
is a value that depends on all the input weights, and these are not
known if we only have a sample. As described in Section
\ref{sec:main-results}, to get confidence bounds out of Theorem
\ref{thm:tail}, it suffices if we based on our sample can compute a
probabilistic upper bound $\ol t_{\max}^\tau$ on the upper bound
$t_{\max}$ such that
\begin{equation}\label{eq:upper-tmax}
\ol t_{\max}^{\,\tau}\geq t_{\max}
\end{equation}
with probability at least $1-P$.
For better confidence lower bounds, we want $\ol t_{\max}$ to be small.
\begin{theorem}\label{thm:good-upper}
Define $\delta^\downarrow_{\leq\tau}$ and $\delta^\uparrow_{\leq\tau}$
such that
$\wp(k_{\leq\tau}/(1-\delta^\downarrow_{\leq\tau}),\,\delta^\downarrow_{\leq\tau}))=P$
and
$\wp(k_{\leq\tau}/(1+\delta^\uparrow_{\leq\tau}),\,\delta^\uparrow_{\leq\tau}))=P$.
Let
\[\ol{t}_{\max}^{\,\tau}=\frac{1+\delta^\uparrow_{\leq\tau}}{1-\delta^\downarrow_{\leq\tau}}\,\tau.\]
Then $\Pr[\ol t^{\,\tau}_{\max}<t_{\max}]\leq P$
\end{theorem}
\begin{proof}
Our first step will be to compute
a probabilistic upper bound $\ol{w}_{\leq\tau}^{\,\tau}$ on $w_{\leq\tau}$
such that 
\begin{equation}\label{eq:upper-weight}
\ol{w}_{\leq\tau}^{\,\tau}\geq w_{\leq\tau}.
\end{equation}
with probability at least $1-P$.  We are going to define
$\ol{w}_{\leq t}^{\,t}$ for any possible threshold $t$ as a function of
only the values of $t$ and $k_{\leq t}$. 
We define $\ol \mu_{\leq t}^{\,t}=k_{\leq t}/(1-\delta^\downarrow_{\leq t})$,
and 
\[\ol w_{\leq t}=t\, \ol \mu_{\leq t}^{\,t}=t\,k_{\leq t}/(1-\delta^\downarrow_{\leq t}).\]  
The lemma below states that $\ol w_{\leq\tau}^{\,\tau}$ does give us the desired probabilistic upper bound on $w_{\leq\tau}$.
\begin{lemma}\label{lem:ulw} For the random priority threshold $\tau$, 
the probability that $w_{\leq\tau}>\ol w_{\leq\tau}^{\,\tau}$ is at most
$P$, so \req{eq:upper-weight} holds true with probability at least $1-P$.
\end{lemma} 
\begin{proof} 
For any given set of input weights consider a threshold $t$
such that $\ol w_{\leq t}^{\,t}\leq w_{\leq t}$. We claim 
that the random priority
threshold $\tau$ is expected no smaller than $t$.
Note that $\tau< t$ if
and only if $|I^{t}_{\leq t}|\leq k_{\leq t}$.
Since $\ol w_{\leq t}^{\,t}\leq w_{\leq t}$, we have $\E[|I^{t}_{\leq t}|]=
w_{\leq t}/t\geq \ol\mu_{\leq t}^{\,t}$. Moreover, $k_{\leq t}=(1-\delta^\downarrow_{\leq t})\ol\mu_{\leq t}^{\,t}$, so
\begin{align*}
\Pr[\tau\leq t]&=\Pr[|I^{\,t}_{\leq t}|\leq
  k_{\leq t}]=\Pr[|I^{t}_{\leq t}|\leq
  (1-\delta^\downarrow_{\leq t})\ol\mu_{\leq t}^{\,t}]\\&\leq\wp(\ol
\mu_{\leq t},\delta^\downarrow_{\leq t})=\wp(k_{\leq t}/(1-\delta^\downarrow_{\leq t}),\delta^\downarrow_{\leq t})\leq P.
\end{align*}
Let $t^+$ be the maximal value such that $\ol w_{\leq t^+}^{t^+}\leq
w_{\leq t^+}$. The probability that $\tau\leq t^+$ is at most $P$, and if
$\tau>t^+$, then $\ol w_{\leq\tau}^{\,\tau}> w_{\leq\tau}$.
If the maximal value $t^+$ does not exist, we define instead $t^+$ as the
limit where $\ol w_{\leq t}^{\,t}\leq w_{\leq t}$ for $t< t+$ while 
$\ol w_{\leq t^+}^{t^+}>w_{\leq t^+}$. The probability that $\tau< t^+$ is at most $P$,
and if  $\tau\geq t^+$, we again have $\ol w_{\leq\tau}^{\,\tau}> w_{\leq\tau}$.
\end{proof}
\begin{lemma}\label{lem:tmax-nice} For any threshold $t$, 
$t_{\max}\leq (1+\delta^\uparrow_{\leq t})w_{\leq t}/k_{\leq t}$.
\end{lemma}
\begin{proof}
Consider any two thresholds $t$ and $T$. Then
$|I_{\leq t}^T|$ is the number of weights below $t$ with threshold above $T$,
and $\tau\geq T$ implies $|I_{\leq t}^T|\geq k_{\leq t}$.
Let $T=(1+\delta^\uparrow_{\leq t})w_{\leq t}/k_{\leq t}$. Then
$\E[|I^T_{\leq t}|]\leq w_{\leq t}/T=k_{\leq t}/(1+\delta^\uparrow_{\leq t})$ with
equality if $T\geq t$. It follows that
\[\Pr[\tau\geq T]\leq \Pr[|I^T_{\leq t}|\geq k_{\leq t}]\leq P.\]
Hence $T\geq t_{\max}$.
\end{proof}
Assuming \req{eq:upper-weight}, with $t=\tau$, we get
\[
t_{\max}\leq (1+\delta^\uparrow_{\leq\tau})w_{\leq\tau}/k_{\leq\tau}
\leq (1+\delta^\uparrow_{\leq\tau})\ol w_{\leq\tau}/k_{\leq\tau}
=(1+\delta^\uparrow_{\leq\tau})\tau/(1-\delta^\downarrow_{\leq t}).\]
By Lemma \ref{lem:ulw}, \req{eq:upper-weight} holds true with
probability $1-P$. This completes the proof of Theorem \ref{thm:good-upper}.
\end{proof}

\subsection{What to trust, and what to expect}\label{sec:heavy-tail}
All our theorems about confidence intervals are trustworthy in the sense that they hold true for any set of
input weights. We will now discuss what to expect if the
input follows a reasonable distribution. As we shall formalize below,
we expect a typical priority sample to 
consist of only a few large weights above the
priority threshold, and a majority of small weights that are significantly
smaller than the priority threshold.
This will imply that our estimated threshold upper bound $\ol t_{\max}^{\,\tau}$ is
very close to the priority threshold $\tau$.

This view has consequences for what we would consider worth
optimizing for in our confidence intervals, e.g., one could try getting better confidence intervals
for cases where the sampled items have weight below but close to the threshold
(information that is currently ignored, and not even contained in the
sample), but since we do not expect many such items, we do not
optimize for this case. 

As our formal model, we assume that each weight $w_i$ is drawn independently
from a Pareto distributions, that is, for a positive real parameter $\alpha=\Omega(1)$, and
any real $x\geq 1$, we have the survival function
\begin{equation}\label{eq:survival}
\ol F(x)=\Pr[w_i\geq x]=1/x^{\alpha}.
\end{equation}
Then all weights are at least $1$.  For $\alpha\rightarrow\infty$, all
weights are 1. As $\alpha$ decreases, we get more heavy weights. The
mean is infinite for $\alpha\leq 1$, and the variance is infinite for
$\alpha\leq 2$.  The probability density function $f$ is the
derivative of $1-\ol F(x)$, so
\begin{equation}\label{eq:density}
f(x)=\alpha/x^{\alpha+1}.
\end{equation}
We are going to use $n$ such input weights as input to a priority
sample of size $k$, where $1\ll k\ll n$. The priority sampling events
assigns priorities $q_i=w_i/r_i$, $r_i\in U(0,1)$ to each item, and in
the analysis, we will study these weights and priorities relative to
any given threshold $t$. For any such given threshold $t$, we assume
that our error probability function $\wp$ from \req{eq:p} holds for
the number of samples $i$, $q_i>t$, for the combined event where we
first assign the weights $w_i$ independently and second assign the
$h_i$ and hence the $q_i$ based on a hash of each $i$.

We assume that the number of priority samples $k$ is so large that
for some small error $\eps=o(1)$ and target error probability $P$, we have
\begin{equation}\label{eq:large-k}
\wp(\Omega(k),O(\eps))=o(P).
\end{equation}
The basic idea is that the number $k$ of samples is so large,
that we do not expect significant errors for the sample as a whole. However,
we might still have significant errors in estimation of small
subsets. More precisely, our analysis will imply the following result.
\begin{theorem}\label{thm:expected-behavior}
Let $t_k$ be the threshold leading to an expected number of $k$ samples, and
let $\tau$ be the actual priority threshold, that is,
the $(k+1)$th largest priority when all random choices are made.
Then, with probability $1-o(P)$, we have $\tau=(1\pm O(\eps))t_k=(1\pm o(1))t_k$. Moreover,
we have $k_{\leq\tau}=\Omega(k)$ small weight samples, most
of which are from weights below $\tau/2$.
\end{theorem}
\begin{proof}
As a first simple observation, since all weights at least $1$, the 
expected number of samples with threshold $t$ is at least $n/t$. It
follows that $k\geq n/t_k$, hence that
\begin{equation}\label{tk-large}
t_k\geq n/k=\omega(1).
\end{equation}
The following analysis is for an arbitrary threshold $t$, not just
$t=t_k$. We want to study the expected number of
large weights $w_i>t$ that are sampled for sure, and the expected
number of small weight samples $w_i\leq t< q_i$. By linearity of expectation,
this is $n$ times the probability of these events for any given item $i$.
By \req{eq:survival}, $\Pr[w_i>t]=1/t^\alpha$. Using the probability density
function \req{eq:density}, we get
\begin{align}
\Pr[w_i\leq t<q_i]&=\int_{1}^{\,t} f(x)\cdot  x/t\; dx\\
&=\alpha/t\cdot  \int_{1}^{\,t} 1/x^\alpha\; dx\label{eq:int}\\
&=\frac{\alpha/t}{1-\alpha}\left[x^{1-\alpha}\right]_{1}^{\,t}\\
&=\frac{\alpha}{1-\alpha}(t^{1-\alpha}-1)/t \label{eq:not-1/t}\\
&=\frac{\alpha}{1-\alpha}(1/t^{\alpha}-1/t)\label{eq:not-1}
\end{align}
This should be compared with the probability
of a large weight sample $w_i\geq t$ which was $1/t^\alpha$. For $\alpha<1$,
the low weight sample probability is $\Omega(1/t^\alpha)$, and
for $\alpha>1/2$ we start expecting more low weight samples than
large weights. For $\alpha>1$, the sampled weights dominate
in that we expect more than $1/t$ if them.

Above we assumed $\alpha\neq 1$. For $\alpha=1$, continuing from
\req{eq:int}, we get 
\begin{align*}
\Pr[w_i\leq t, q_i>t]&=1/t\cdot [\ln x]_{1}^{\,t}\\
&=1/t\cdot\ln t
\end{align*}
which means that the small weight samples are dominant by a factor of
$\ln t$ for $\alpha=1$. For continuity, it is easily verified that \req{eq:not-1} also
converges to $(\ln t)/t$ for $\alpha\rightarrow 1$. For simplicity, we
assume below that $\alpha\neq 1$.

The total expected number of samples is 
\[s_\alpha(t)=\frac{1/t^{\alpha}-\alpha/t}{1-\alpha}\]
By definition, $s_\alpha(t_k)=k$. Define $t_k^+>t_k$ such that
$s_\alpha(t_k^+)=k/(1+\eps)$ where $\eps=o(1)$ is the error from
\req{eq:large-k}. To get priority threshold $\tau\geq t^+_k$, we need
at least $k+1$ samples $t^+_k$. By \req{eq:large-k}, this happens
with probability $o(P)$. If $\alpha<1$, then 
$t_k^+<(1+\eps)^{1/\alpha}\,t_k$, and if $\alpha>0$, 
$t_k^+<(1+\eps)t_k$.
Since $\alpha=\Omega(1)$, we conclude in both cases that 
\[t_k^+<(1+O(\eps))t_k.\]
A symmetric argument shows that $\tau\geq (1-O(\eps))t_k$ with probability 
$1-o(P)$. Thus
$\tau=(1\pm O(\eps))t_k$ with probability $1-o(P)$.

Next we need to argue that with probability $1-o(P)$,
the number $k_{\leq\tau}$ of sampled small weights in the priority sample
is $\Omega(k)$. We may assume that $\tau\leq t^+_k$, so $k_{\leq\tau}$
is at least as big as the number of
sampled small weights with threshold $t^+_k$. By definition of $t^+_k$,
the expected number of
samples with threshold $t^+_k$ is $k/(1+\eps)$, and we know for any
given threshold, that the expected number of small weight samples is
at least a constant fraction of the expected total, so we expect
$\Omega(k)$ small weight samples with threshold $t^+_k$. By
\req{eq:large-k}, this implies that their actual number is
$\Omega(k)$ with probability $1-o(P)$. Thus $k_{\leq\tau}=\Omega(k)$
with probability $1-o(P)$.

Finally, among the sampled small weights $w_i\leq\tau<q_i$, we want
to see what fraction is below $\tau/2$. As usual, in our analysis,
we first consider given thresholds rather than the variable priority threshold.
Generally, for given thresholds $t_0\leq t_1$, and any given value of 
$w_i\leq t_0$, 
\[\Pr[q_i>t_1|w_i]=\Pr[q_i>t_0| w_i]\cdot t_0/t_1.\]
Hence
\[\Pr[w_i\leq t_0,q_i>t_1]=\Pr[w_i\leq t_0<q_0]\cdot t_0/t_1=
\frac{\alpha}{1-\alpha}(t_0^{1-\alpha}-1)/t_1.\]
For $\alpha=\Omega(1)$ and $t_1\geq t_0=\omega(1)$, we get 
\[\Pr[w_i\leq t_0,q_i>t_1]>(t_0/t_1)^{1-\Omega(1)}\cdot \Pr[w_i\leq t_1,q_i>t_1].\]
We know that with probability $1-o(P)$, that $t_k^-\leq \tau\leq t_k^+$
where $t_k^+=(1+o(1)t_k^-$. A weight $w_i$ is in the priority sample if $q_i>\tau$.
We say weight $w_i$ is over-sampled if $q_i\geq t_k^-$
and under-sampled if $q_i\geq t_k^+$. The over-sampled weights $w_i\leq t^+_k$ include all weights $w_i\leq \tau$ in the priority
sample, and the under-sampled weights $w_i\leq t^-_k/2$
are all included among the weights $w_i\leq \tau/2$ in the priority sample. 

The expected number of over-sampled weights $w_i\leq t^+_k$ is 
\begin{equation}\label{eq:medium}
n\cdot \Pr[w_i\leq t_k^+,q_i>t_k^-]\leq n\cdot 
\frac{\alpha}{1-\alpha}((t_k^+)^{1-\alpha}-1)/t_k^-.
\end{equation}
while the expected number of under-sampled weights $w_i\leq t^-_k/2$
is 
\begin{equation}\label{eq:tiny}
n\cdot \Pr[w_i\leq t_k^+,q_i>t_k^-]= n\cdot 
\frac{\alpha}{1-\alpha}((t_k^-/2)^{1-\alpha}-1)/t_k^+.
\end{equation}
With $\alpha=\Omega(1)$, \req{eq:tiny} is within a factor
$1/2+\Omega(1)$ of \req{eq:medium}. Moreover, from our analysis of
$k_{\leq\tau}$, we know that \req{eq:medium} and hence \req{eq:tiny} is
$\Omega(k)$. It follows from \req{eq:large-k} that with probability
$1-o(P)$, the expected bounds \req{eq:medium} and \req{eq:tiny} end up both 
satisfied within a factor $1\pm o(1)$.
Then, among the priority sampled small weights $w_i\leq\tau$, at
least half have weight below $\tau/2$.
\end{proof}
We now return to our confidence bounds for large errors.
By Theorem \ref{thm:expected-behavior}, with 
probability $1-o(P)$, we get $k_{\leq \tau}=\Omega(k)$. Hence 
by \req{eq:large-k}, we get 
$\delta^\uparrow_{\leq\tau},\delta^\downarrow_{\leq\tau}=O(\eps)=o(1)$
in Theorem \ref{thm:good-upper}, 
so
\[\ol{t}_{\max}^{\,\tau}=(1+o(1))\tau\]
and then $\ol{\ell}=1$ in \req{eq:oll}.

\subsection{Histogram similarity}\label{sec:histogram}

We will now discuss estimators for the similarity of weighted sets.
First consider the simple case where each key has a unique weight. 
The similarity is then just the total weight of the intersection divided by the  weight
of the union, and we estimate these two quantities independently.

As in the bottom-$k$ sample for unweighted items, we note that given the size-$k$ priority
sample of two sets $A$ and $B$, we can easily construct the size-$k$
priority sample of their union, and identify which of these samples
come from the intersection. Our analysis for subset sums now applies
directly.

In the case of histogram similarity, it is
natural to allow the same item to have different weights in different
sets.  More specifically, allowing zero weights, every possible item
has a weight in each set.  For the similarity we take the sum of the
minimum weight for each item, and divide it by the sum of the maximum
weight for each item.  
Formally, we are considering two sets $A$ and $B$. Item $i$ has
weight $w_i^{\,A}$ in $A$ and weight $w_i^{\,B}$ in $B$.  Let
$w_i^{\,\max}=\max\{w_i^{\,A},w_i^{\,B}\}$ and
$w_i^{\,\min}=\min\{w_i^{\,A},w_i^{\,B}\}$. The histogram similarity is $w^{\,\min}/w^{\,\max}=(\sum_i w_i^{\,\min})/(\sum_i w_i^{\,\max})$.

This would seem a perfect application of our fractional subsets
with $w_i=w_i^{\,\max}$ and $f_i=w_i^{\,\min}/w_i^{\,\max}$. The issue
is as follows. From our priority samples over the $w_i^{\,A}$ and $w_i^{\,B}$, 
we can easily get the priority sample for the $w_i=w_i^{\,\max}$. However,
for the items $i$ sampled, we would typically not have a sample with
$w_i^{\,\min}$, and then we cannot compute $f_i$.

Our solution is to keep the instances of an item $i$ in $A$ and $B$
separate as twins $i^{\,A}$ and $i^{\,B}$ with priorities $q_i^{\,A}=w_i^{\,A}/h_i$ and
$q_i^{\,B}=w_i^{\,B}/h_i$. Note that it is the same hash value
$h_i$ we use to determine these two priorities. If $w_i^{\,A}=w_i^{\,B}$, we get
$q_i^{\,A}=q_i^{\,B}$, and then we break the tie in favor of $i^{\,A}$. 
The priority sample for the union $A\cup B$
consists of the split items with the $k$ highest priorities, and the priority threshold $\tau$ is the
$k+1$ biggest among all priorities. Estimation is done as usual: for $C\in\{A,B\}$, if $i^{\,C}$ is
sampled, $\widehat w_i^{\,C}=\max\{w_i^{\,C},\tau\}$. The important
point here is the interpretation of the results. If $w_i^{\,A}\geq w_i^{\,B}$,
then the priority of $i^{\,A}$ is higher than that of $i^{\,B}$. Thus, in our sample,
when we see an item $i^{\,{C}}$, $C\in \{A,B\}$,  we count it for the 
union $\widehat w^{\,\max}$ if it is not preceded by its twin; otherwise we
count it for the intersection $\widehat w^{\,\min}$.

The resulting estimators $\widehat w^{\,\min}$ and $\widehat w^{\,\max}$ will
no longer be unbiased even with truly random hashing. To see this, note that
with sample size $k=1$, we always get $\widehat w^{\,\min}=0.$
However, for our concentration bounds, we only lose a constant factor.
The point is that the current analysis is
using union bounds over threshold sampling events, using
the fact that each hash value $h_i$ contributes at most 1 to
the number of items with priorities above a given threshold $t$. Now
$h_i$ affects at most $2$ twins, but this is OK since all we really need is that
the contribution of each random variable is bounded by a constant. The
only effect on Theorem \ref{thm:main} is that we replace the
relative error $6\,\delta(f w_{<\tau}/(3\tau),P)$ with
$6\,\delta(f w_{<\tau}/(6\tau),P)$.

\end{document}